\documentclass[reqno,12pt,letterpaper]{amsart}
\usepackage{amsmath,amssymb,amsthm,graphicx,mathrsfs,url}
\usepackage[usenames,dvipsnames]{color}
\usepackage[colorlinks=true,linkcolor=Red,citecolor=Green]{hyperref}
\usepackage{amsxtra}
\newcommand{\vertiii}[1]{{\left\vert\kern-0.25ex\left\vert\kern-0.25ex\left\vert #1 
    \right\vert\kern-0.25ex\right\vert\kern-0.25ex\right\vert}}
\usepackage{tikz}
\usetikzlibrary{arrows.meta}
\usepackage{wasysym} 
\usepackage{graphicx}
\usepackage{subcaption}
\def\arXiv#1{\href{http://arxiv.org/abs/#1}{arXiv:#1}}
\usepackage{array}
\newcolumntype{P}[1]{>{\centering\arraybackslash}m{#1}}

\def\?[#1]{\textbf{[#1]}\marginpar{\Large{\textbf{??}}}}

\usepackage[normalem]{ulem}

\def\smallsection#1{\smallskip\noindent\textbf{#1}.}
\let\epsilon=\varepsilon 

\newcommand{\RR}{{\mathbb R}}

\newcommand{\CC}{{\mathbb C}}

\newcommand{\ZZ}{{\mathbb Z}}

\newtheorem{theo}{Theorem}
\newtheorem{que}{Question}
\newtheorem{prop}[theo]{Proposition}
\newtheorem{defi}[theo]{Definition}

\newtheorem{assumption}[theo]{Assumption}
\newtheorem{lemm}[theo]{Lemma} 

\newtheorem{rem}[theo]{Remark}
\numberwithin{equation}{section}

\DeclareMathOperator{\Spec}{Spec}

\let\Im=\Imag

\let\Re=\Real

\DeclareMathOperator{\supp}{supp}

\DeclareMathOperator{\tr}{tr}

\def\indic{\operatorname{1\hskip-2.75pt\relax l}}
\author{Simon Becker}\thanks{S.B. is the author to whom any correspondence should be addressed}
\email{simon.becker@math.ethz.ch}
\address{ETH Zurich, 
Institute for Mathematical Research, 
Raemistrasse 101, 8092 Zurich, 
Switzerland}
\author{Izak Oltman}
\email{ioltman@berkeley.edu}
\address{University of California, Department of Mathematics, Berkeley, CA 94720, USA}
\author{Martin Vogel}
\email{vogel@math.unistra.fr}
\address{Université de Strasbourg,
Institut de Recherche Mathématique Avancée,
67084 Strasbourg Cedex, France}
\begin{document}
\title[Disordered TBG]{Magic angle (in)stability and mobility edges in disordered Chern insulators} 
\begin{abstract}
Why do experiments only observe one magic angle in twisted bilayer graphene, despite standard models like the chiral limit of the Bistritzer-MacDonald Hamiltonian predicting an infinite number? In this article, we explore the relative stability of larger magic angles compared to smaller ones. Specifically, we analyze how disorder impacts these angles as described by the Bistritzer-MacDonald Hamiltonian in the chiral limit. Changing focus, we investigate the topological and transport properties of a specific magic angle under disorder. We identify a mobility edge near the flat band energy for small disorder, showing that this mobility edge persists even when all Chern numbers are zero. This persistence is attributed to the system's $C_{2z}T$ symmetry, which enables non-trivial sublattice transport. Notably, this effect remains robust beyond the chiral limit and near perfect magic angles, aligning with experimental observations.
\end{abstract}
\maketitle
\section{Introduction}
Twisted bilayer graphene (TBG) is a highly tunable material that. As predicted by the standard Hamiltonian that describes its band structure, the Bistritzer-MacDonald Hamiltonian \cite{BM11}, TBG exhibits nearly flat bands at specific twisting angles, known as \emph{magic angles}\cite{magic}. These magic angles are of great interest due to the intriguing topological properties of the associated Bloch bundles, which help explain the quantum anomalous Hall effect in TBG, as well as other compelling many-body phenomena. This article aims to initiate the first study on the robustness of these effects in the presence of disorder.

We begin by examining the stability of the magic angles. Although the Bistritzer-MacDonald Hamiltonian predicts an infinite number of magic angles \cite{FBM}, only the largest of these angles has been experimentally observed. This raises an important question:
\begin{que}
\label{que0}
Why is the largest magic angle more robust than the smaller magic angles predicted by the chiral limit of the Bistritzer-MacDonald Hamiltonian? Can we quantify this stability?
\end{que}

The Bistritzer-MacDonald Hamiltonian, which we thoroughly introduce in Subsection \ref{subsec:chiral_limit}, is the standard effective one-particle Hamiltonian used to describe twisted bilayer graphene—a material consisting of two graphene sheets stacked and twisted relative to each other. By focusing on the chiral limit of this Hamiltonian, we identify a discrete and infinite set of magic angles where perfectly flat bands appear at zero energy, as defined in Definition \ref{defi:magic_angles}. This study provides a possible explanation for why only the largest of these magic angles has been observed experimentally. We discuss our results on Question 1 in Subsection \ref{subsec:Q1}.

In the second part of this article, we examine the effects of disorder within the chiral limit, particularly its interaction with the flat bands at a \emph{fixed magic angle}. Unlike the first part, which addresses the entire set of magic angles, this section focuses on the impact of disorder on the spectral and dynamical properties of the Hamiltonian at a specific magic angle. Here, we explore how disorder influences the topological properties of Bloch bundles, the spectral characteristics of the Hamiltonian, and the transport properties governed by the underlying Schrödinger equation when disorder is introduced at a magic angle near zero energy. This raises the key question:

\begin{que}
\label{que1}
How do the topological properties of the Bloch bundles, the spectral types of the Hamiltonian, and the transport properties of the underlying Schrödinger equation change when disorder is added to the Hamiltonian at a fixed magic angle close to zero energy?
\end{que}
We survey our results on this question in Subsection \ref{sec:Q2}.
In quantum systems, disorder-induced dynamical localization is a well-established phenomenon where spatially localized wavepackets exhibit minimal diffusion over time. Although the mechanisms behind localization are well understood, understanding its opposite, diffusive behavior in disordered systems, remains limited to specific examples \cite{An,AW,BH,GKS,Kl98,JSS}.

It is widely believed \cite{AALR} that many two-dimensional quantum systems, even under mild disorder, predominantly exhibit localization, as conjectured in Problem 2 on Simon's list of open problems for Schrödinger operators \cite{Si00}. One aim of this article is to highlight an exceptional class of materials that challenges this belief, \emph{Chern insulators}. These materials possess Bloch bundles with non-zero Chern numbers, even without external magnetic fields that typically break the time-reversal symmetry of the Hamiltonian \cite{Li21}.

In the context of disordered magic angle twisted bilayer graphene, we demonstrate that wavepackets localized near zero energy exhibit, in a suitable sense, ballistic time evolution. Our argument adapts a method by Germinet, Klein, and Schenker, who showed a form of delocalization for the Landau Hamiltonian \cite{GKS}. The physical intuition behind delocalization in a magnetic Hamiltonian is straightforward: the Landau Hamiltonian exhibits non-zero Hall conductivity at each Landau level, a topological invariant characterized by Chern numbers, that remains stable under minor disorder. The existence of spectral gaps between Landau levels prevents strong localization across the spectrum.

The flat bands in twisted bilayer graphene are somewhat analogous to Landau levels, with the key difference being that no magnetic field is involved, and the net Chern number is zero. At the first magic angle, the two flat bands at zero energy correspond to a Bloch bundle with a total Chern number of zero. However, each flat band individually gives rise to bundles with non-zero Chern numbers of ±1, enabling an anomalous quantum Hall effect when the TBG substrate is aligned with hexagonal boron nitride (hBN) \cite{Li21}. Mathematically, the effect of this alignment is modeled by adding an effective mass term to the Hamiltonian, splitting the two flat bands. Additionally, it has been shown that the flat bands are separated by a gap from the rest of the spectrum \cite{bhz2,bhz4}.

We also establish a localized regime using the multi-scale analysis developed by Germinet and Klein \cite{GK01,GK02}. The primary challenge here is accommodating a sufficiently large class of random perturbations, which requires extending the estimate on the number of eigenvalues (NE) and the Wegner estimate (W) to our matrix-valued differential operator, the Bistritzer-MacDonald Hamiltonian, which we introduce next.

\subsection{Chiral limit of Bistritzer-MacDonald Hamiltonian}
\label{subsec:chiral_limit}
In this subsection, we briefly review the key mathematical aspects of magic angles within the one-particle framework, providing the necessary context to present our results. For a more detailed mathematical treatment, we refer the reader to \cite[Section 3]{dynip}, which elaborates on the concepts summarized below.

The chiral limit of the massive Bistritzer-MacDonald (BM) Hamiltonian for twisted bilayer graphene is the periodic Hamiltonian $H(m,\alpha)$ acting on $L^2(\CC;\CC^4)$ with domain defined by the Sobolev space $H^1(\CC;\CC^4)$
\begin{equation}
\label{eq:Hamiltonian}
 H(m,\alpha) = \begin{pmatrix} m I_2 & D(\alpha)^* \\ D(\alpha) & -m  I_2\end{pmatrix} \text{ with } D(\alpha) = \begin{pmatrix} 2 D_{\bar z} & \alpha U(z) \\ \alpha U(-z) & 2D_{\bar z} \end{pmatrix}
 \end{equation}
 and $I_n$ is the $n \times n$ identity matrix. 
Here  $D_{\bar z} = -i \partial_{\bar z}$ , $\alpha \in \mathbb C\setminus\{0\}$ is an effective parameter that is inversely proportional to the physical twist angle $\theta$ and $m \ge 0$ is an effective mass parameter. As mentioned in the introduction, the mass parameter models the effect of aligning the twisted bilayer graphene (TBG) with other materials such as hexagonal boron nitride, which is crucial for observing the anomalous quantum Hall effect \cite{Li21}.
The Hamiltonian in \eqref{eq:Hamiltonian} with a positive mass parameter $m >0$ also serves as a model for twisted transition metal dichalcogenides (TMDs) \cite{CRQ23}. However, in the context of TBG without any auxiliary substrate, the model typically assumes $m=0$.

Let $\Gamma:=4\pi i \omega (\ZZ \oplus \omega \ZZ)$ be a triangular lattice with $\omega = e^{2\pi i /3}$. The tunnelling potentials $U$ are $\Gamma$-periodic functions that respect the symmetries
\begin{equation}
\label{eq:symmetries}
U(z+\mathbf a)=\bar\omega^{a_1+a_2}U(z), \quad U(\omega z) = \omega U(z), \quad \overline{U(z)}=U(\bar z)
\end{equation} 
for $\mathbf a = 4\pi i a_1 \omega/3+ 4\pi i a_2 \omega^2/3$ with $a_i \in \ZZ$, i.e. $\mathbf a \in \Gamma_3 :=\Gamma/3.$

Given that Hamiltonian is periodic with respect to $\Gamma$, we can apply the Bloch-Floquet decomposition of $H$ and equivalently study the family 
\[ H_k:=e^{-i \Re(z \bar k)} H e^{i \Re(z \bar k)}: H^1(\CC/\Gamma;\CC^4) \subset L^2(\CC/\Gamma;\CC^4) \to L^2(\CC/\Gamma;\CC^4)\]
where
\begin{equation}
\label{eq:Bloch_Hamiltonian}
 H_k(m,\alpha):= \begin{pmatrix} m I_2 & D(\alpha)^* + \bar k I_2 \\ D(\alpha)+k I_2 & -m I_2 \end{pmatrix} \end{equation}
with quasi-momentum $k \in \mathbb C,$ see also \cite[(2.11)]{beta}. The range of the individual Bloch eigenvalues over all $k$ are called bands and we can order them as follows
\[ ...\le E_{-2}(k)\le E_{-1}(k) \le -m \le 0 \le m \le E_1(k)\le E_2(k)\le ... \ .\]
The central objects in the one-particle picture of twisted bilayer graphene are the \emph{magic angles} at which the Hamiltonian exhibits perfectly flat bands $E_{\pm 1},$ i.e. $E_{\pm 1}$ do in fact not depend on $k$.

\begin{defi}[Magic angles] 
\label{defi:magic_angles} 
We say that $\alpha \in \mathbb C\setminus \{0\}$ is \emph{magic} if and only if the Bloch-Floquet transformed Hamiltonian with mass parameter $m \ge 0$ exhibits a flat band at energy $\pm m$. In short, $\alpha \in \CC \setminus \{0\}$ is magic if and only if
\begin{equation}
\label{eq:flat_band}
 \pm m \in\bigcap_{k \in \CC} \Spec_{L^2(\CC/\Gamma)}(H_{k}(m,\alpha))
 \end{equation}
 which is equivalent to saying that $E_{\pm 1}(k)=\pm m$ for all $k \in \CC.$
Here, $\Spec_{X}(S)$ denotes the spectrum of the linear operator $S$ on the Hilbert space $X$ on a suitable dense domain, where as before $H_{k}(m,\alpha):H^1(\CC/\Gamma;\CC^2) \to L^2(\CC/\Gamma;\CC^2).$ \end{defi}
 The set of parameters  $\alpha \in \CC$ for which there exists a flat band at energy $\pm m$ that we denote by $\mathcal A$, is independent of $m$\footnote{This is an easy consequence of $\Spec H_k(m,\alpha) = \pm \sqrt{ \Spec H_k(0,\alpha)^2 + m^2}$ \cite[(5.66)]{T92}.}. In the following, we shall suppress the mass parameter $m \ge 0$ in the notation when it does not affect the analysis.
 
Away from magic $\alpha$, that is, for $\alpha \notin \mathcal A$, it is known that $\pm m  \in \Spec(H_k(m,\alpha))$ if and only if $k \in \Gamma^*,$ where $\Gamma^*$ is the dual lattice.  
To summarize, we have the following magic-angle criterion:  
\begin{equation}
\label{eq:criterion}
\text{There exists }k \notin \Gamma^*\text{ such that} \pm m  \in \Spec(H_k(m,\alpha))\text{ if and only if  }\alpha \in \mathcal A.
\end{equation}

For the study of magic angles we also introduce a translation operator
\begin{equation}
\label{eq:La}
  \mathscr L_{\mathbf a } w ( z ) := 
\begin{pmatrix}  \omega^{a_1 + a_2}  & 0 \\
0 & 1  \end{pmatrix} 
w( z + \mathbf a ),  \ \ \ \mathbf a \in \Gamma_3  , 
\end{equation}
and a rotation operator $\mathscr Cu(z)=u(\omega z)$ that both commute with the operator $D(\alpha)$ in \eqref{eq:Hamiltonian}. The reason for introducing the above translation operator is that although the Hamiltonian is periodic with respect to $\Gamma$, it also satisfies a translation symmetry with respect to $\Gamma_3$, but with the modified translation operator \eqref{eq:La}, $ \mathscr L_{\mathbf a }  \otimes I_2.$ 

We can then define special invariant subspaces of $D(\alpha)$ for $\ell,p\in \ZZ_3$
\begin{equation}
\label{eq:subspaces}
L^2_{\ell,p}:=\{ u \in L^2(\CC/\Gamma;\CC^2); \mathscr L_{\mathbf a} u(z) = \omega^{\ell} u(z) \text{ and } \mathscr Cu(z) = \bar\omega^p u(z)\}
\end{equation}
and define their direct sum $L^2_{\ell}:=\bigoplus_{p \in \ZZ_3} L^2_{\ell,p}$.

The set $\mathcal A$ of such magic $\alpha$, as in Definition \ref{defi:magic_angles} is characterized by the eigenvalues of a compact operator. 
\begin{theo}{\cite[Theo.$2$]{beta}}
\label{theo:beta}
The parameter $\alpha \in \mathbb C \setminus \{0\}$ is magic, as in Definition \ref{defi:magic_angles} if and only if
\begin{equation}
\label{eq:magic}
 \alpha^{-1} \in \Spec_{L^2_0}(T_k) \text{ with } T_k = (2D_{\bar z}+k)^{-1}\begin{pmatrix} 0&  U(z) \\  U(-z) & 0 \end{pmatrix} \text{ for some }k \notin \Gamma^*,
 \end{equation}
Moreover, the spectrum of $T_k$ is independent of $k \notin \Gamma^*.$
 \end{theo}
 
 To see how $T_k$ enters in the discussion, notice the simple equivalence
\begin{equation}
\label{eq:easy-peasy}
-k \in \Spec(D(\alpha)) \Leftrightarrow -\alpha^{-1} \in \Spec(T_k)
\end{equation}
 which holds for any $\alpha \neq 0$ and $k \notin \Gamma^*.$ In fact the non-trivial part about Theorem \ref{theo:beta} is the rigidity that the spectrum of $T_k$ is independent of $k \notin \Gamma^*.$ What this means is that if $-k \in \Spec(D(\alpha))$ for some $k \notin \Gamma^*$, then $-k \in \Spec(D(\alpha))$ for all $ k \in \CC.$ The latter is then equivalent to the flat band condition in Def. \ref{defi:magic_angles}. 

We now introduce the concept of \emph{generic magic angles}. The term '\emph{generic}' is inspired by \cite[Theo. 3]{bhz4}, which demonstrates that, for a generic (in the Baire sense) choice of tunneling potentials
$U$, all magic angles exhibit specific properties outlined in the following definition.
\begin{defi}[Generic magic angles]
\label{def:generic_mag_angle}
We say that $\alpha \in \mathcal A$ is a simple or two-fold degenerate magic angle if $1/\alpha \in \Spec_{L^2_0}(T_k)$ and $\dim \ker_{L^2_0}(T_k-1/\alpha)=\nu$ with $\nu=1,2$, respectively. We refer to the union of these magic angles as \emph{generic magic angles.}
\end{defi}
\subsection{Magic angle (in)stability}
\label{subsec:Q1}
The first aim of this article is to study Question \ref{que0}, the stability of (the set of) magic angles in the chiral limit. We model this by studying perturbations $W$ of the chiral Hamiltonian \eqref{eq:Hamiltonian} with relative coupling strength $\lambda$, i.e. $\lambda=0$ recovers the unperturbed Hamiltonian \eqref{eq:Hamiltonian}, 
\begin{equation}
\label{eq:Hk}
 \begin{aligned} H_{k,\lambda}(m,\alpha):=\begin{pmatrix} m  I_2 & D_{k,\lambda}(\alpha)^*  \\ D_{k,\lambda}(\alpha) & -m I_2 \end{pmatrix} \text{ with } \\
D_{k,\lambda}(\alpha):=\begin{pmatrix} 2D_{\bar z} +k& \alpha U(z) \\ \alpha U(-z)  & 2D_{\bar z}+k \end{pmatrix} + \lambda \alpha  W \text{ and } W=\begin{pmatrix}  A_+ &  V_+ \\ 
 V_-  &  A_- \ 
\end{pmatrix}, \end{aligned}
\end{equation}
where we introduced bounded linear perturbation operators $A_{\pm},V_{\pm}$. We use the notation $A_{\pm}$ and $V_{\pm}$ to indicate that the perturbations, if we assume that they are multiplication operators, correspond to magnetic potentials, $A_{\pm}$, and tunnelling potentials $V_{\pm}.$ Multiplying the perturbation operators by $\alpha$ reflects that the perturbations act on the same length scale as the original tunnelling potentials, the \emph{moir\'e length scale}. A discussion of this correspondence can be found, for example, in \cite{suppl,magic}.

Our Theorem \ref{theo:perturbation} then provides an upper bound and therefore a \emph{stability bound}, on the shift of magic angles under such perturbations. This bound is more restrictive for large magic angles. This result is significant because the operator $T_k$, whose eigenvalues correspond to the magic angles, is non-normal. As a result, even small perturbations in norm could potentially cause substantial shifts in the spectrum, as discussed in \cite[Theo. 10.2]{ET05}. 
 Conversely, we obtain a lower bound \emph{instability bound} on the shift, showing that even simple rank $1$ perturbations of exponentially small size in the large parameter $1/\vert \theta \vert$ suffice to generate eigenvalues $\theta$ in the spectrum of $T_k$.
\begin{theo}[Instability]
\label{theo:instability}
Let $\alpha \in \mathbb C \setminus \{0\}$ and $k \notin \Gamma^*$, then there exists a rank-$1$ operator $R$ with $\Vert R \Vert = \mathcal O(e^{-c(k) \vert \alpha \vert})$ and $c(k)>0$ independent of $\alpha$ such that $-\alpha^{-1} \in \Spec(T_{k}+R)$.
\end{theo}
This instability result informally states that for $\pm m$ to be in the spectrum of $H_{k,\lambda}(m,\alpha)$, the perturbed Hamiltonian \eqref{eq:Hk} at the twisting angle $\theta \propto \alpha^{-1}$,  a perturbation of size $\vert \lambda \vert = \mathcal O(e^{-c/\vert \theta \vert})$ suffices. 

To see this more clearly, recall that the perturbation of $T_k$ by a rank 1 operator $R$ is, following \eqref{eq:easy-peasy}, equivalent to a rank 1 perturbation $\lambda \alpha W$ of $D(\alpha)$ where $R = (2D_{\bar z}+k)^{-1}\lambda W$, i.e. 
\[ -k \in \Spec(D_{k,\lambda}(\alpha)) \Longleftrightarrow -\alpha^{-1} \in \Spec(T_k +R).\]
Thus, Theorem \ref{theo:instability} shows that if we fix any $k \notin \Gamma^*$ and perturb $D(\alpha)$ by a rank-$1$ operator $W$ with exponentially small coupling parameter $\lambda$,  then we can ensure that $-k \in \Spec(D(\alpha)+\lambda \alpha  W).$ This perhaps at first surprising result, shows that arbitrary eigenvalues $\pm m$, at the flat band energy level, can be easily generated for small twisting angles, i.e.
\[ \pm m \in \Spec(H_{k,\lambda}(m,\alpha)).  \]
If this condition holds for the unperturbed Hamiltonian, i.e. $\pm m \in  \Spec(H_{k}(m,\alpha))$, then, as explained before, this implies that $\alpha$ corresponds to a magic angle. 
There are at least two directions in which it would be desirable to improve the results of Theorem \ref{theo:instability}. The first concerns the nature of admissible perturbations. 
We consider rank 1 perturbations which are easy to treat. However, these may not be the most physically natural examples when compared with potential perturbations.  Such perturbations are covered by the stability bound in Theorem \ref{theo:perturbation}. Potential perturbations are also discussed in more detail using a very different approach in our companion article \cite{bov}.

By our above explanation, Theorem \ref{theo:instability} shows that it is easy to satisfy the magic angle condition \eqref{eq:criterion} of the unperturbed operator $D_k(\alpha)$ for the perturbed operator $D_{k,\lambda}(\alpha)$ in the small limit $\theta \propto 1/\alpha$. Although the magic angle condition \eqref{eq:criterion} implies the existence of perfectly flat bands for the unperturbed Hamiltonian $H(m,\alpha)$ at energy $\pm m$, this may no longer be true for the perturbed Hamiltonian. This should not be too surprising, as the existence of perfectly flat bands is a very special feature that heavily relies on symmetries. Even the full Bistritzer-MacDonald Hamiltonian is not believed to exhibit perfectly flat bands away from the chiral limit studied in this work. However, it may be desirable to consider magic angle criteria other than \eqref{eq:criterion} and analyze their stability under random perturbations.

\subsection{Anderson model and IDS}
\label{sec:Q2}
We now shift our perspective to Question \ref{que1}, focusing on the impact of disorder near a fixed magic angle. One notable consequence of a flat band is the presence of jump discontinuities in the \emph{integrated density of states} (IDS). 
The integrated density of states is defined as follows; see \cite{Sj} and others:
 \begin{defi}
 \label{defi:IDS}
 The \emph{integrated density of states (IDS)} for the energies $E_2>E_1$ and $I = [E_1,E_2]$ is defined by 
 \[ N(I):=\lim_{L \to \infty} \frac{\tr(\indic_{I}(H_{\Lambda_L}(\alpha)))}{\vert \Lambda_L\vert}\]
 with $\Lambda_L=\CC/(L\Gamma)$. $H_{\Lambda_L}(\alpha)$ is the Hamiltonian \eqref{eq:Hamiltonian} with periodic boundary conditions, i.e. $H_{\Lambda_L}(\alpha):H^1(\Lambda_L) \subset L^2({\Lambda_L}) \to L^2({\Lambda_L}).$
 \end{defi}
 Here, $\Lambda_L=\CC/(L\Gamma)$ is the fundamental domain of the lattice $L\Gamma$ that we, with some abuse of notation, identify with its representative centered at the origin. 
 For ergodic random operators, the almost sure existence of this limit is shown using the subadditive ergodic theorem; see, for instance, \cite[Sec. 7.3]{Kirsch}.
Alternatively, one may define for $f \in C_c^{\infty}(\RR)$ the regularized trace 
\begin{equation}
\label{eq:reg_trace}
\widetilde \tr(f(H(\alpha))) = \lim_{L \to \infty} \frac{\tr(\indic_{{\Lambda_L}} f(H(\alpha)))}{\vert \Lambda_L\vert}.
\end{equation}
By Riesz's representation theorem, one has that $$\widetilde \tr(f(H(\alpha))) = \int_{\RR} f(\lambda) \ d\rho(\lambda),$$ where $\rho$ is the \emph{density of states (DOS) measure} of $H(\alpha).$ This way, $N(I) = \int_I \ d\rho(\lambda).$
\begin{rem} 
For Schr\"odinger operators it is common to consider Dirichlet approximations of the finite-size truncation in the density of states. It is known that Dirac operators, as they are first-order operators, generally do not have any self-adjoint Dirichlet realizations. However, self-adjoint Neumann-type boundary conditions are possible, see \cite{Berry} and, for instance, the introduction of \cite{SV19} for a mathematical discussion. The independence of the definition of the IDS of the boundary conditions can then be shown using spectral shift function techniques if the operator contains a gap in the spectrum, see for instance the work by Nakamura \cite{N} on Schr\"odinger operators.
\end{rem}
For the periodic Hamiltonian $H(\alpha)$ with Bloch operators $H_k(\alpha)$ we have \cite[(1.29)]{Sj}
 \[ N(I) = \int_{\mathbb C/\Gamma^*} \Bigg( k \mapsto   \sum_{\lambda \in \Spec_{L^2(\CC/\Gamma)}(H_k(\alpha))} \indic_I(\lambda) \Bigg) \ \frac{dk}{4\pi^2}.\]

 In particular, a periodic Hamiltonian that has a flat band, such as \eqref{eq:Hamiltonian} for magic $\alpha$, at energy $E$ possesses a jump discontinuity in the IDS at $E$. In particular, the Lebesgue decomposition of $\rho$ has a pure point contribution at $E.$  As a consequence, if we define the associated cumulative distribution function $N_{E_0}: (E_0 ,\infty) \to \mathbb R$ by $N_{E_0}(E):=N([E_0,E])$, then this function is monotonically increasing and right-continuous (càdlàg). At a magic angle, the function $N_{E_0}$ for $E_0<\pm m$ has a jump discontinuity at $E= \pm m.$


Let $\alpha \in \mathcal A$ be a generic magic angle, as in Def. \ref{def:generic_mag_angle}, then we define the energy gap between the flat bands with $m=0$ and the rest of the spectrum by
\begin{equation}
\label{eq:gap}
E_{\operatorname{gap}}(\alpha):=\inf_{\lambda \in \Spec(H(m=0,\alpha)^2)\setminus \{0\}} \sqrt{\lambda}>0.
\end{equation} 
The existence of a spectral gap follows from \cite[Theo.$2$]{bhz2} for simple and from \cite[Theo.$4$]{bhz4} for two-fold degenerate magic angles and thus holds for all generic magic angles. We illustrate this in Figure \ref{fig:gap}. 
To summarize, for $\alpha$ a generic magic angle, the following union of intervals is in the resolvent set of the Hamiltonian $H(m,\alpha)$
\begin{equation}
\label{eq:spectral_gap}
\Big(-\sqrt{E_{\operatorname{gap}}(\alpha)^2+m^2},-m\Big) \cup \big(-m,m\big) \cup \Big(m,\sqrt{E_{\operatorname{gap}}(\alpha)^2+m^2}\Big) \subset \mathbb R \setminus \Spec(H(m,\alpha)).
\end{equation}
Let $P_X$ be the orthogonal projection onto a closed subspace $X$. For $\alpha \in \mathcal A$ generic, as in Def. \ref{def:generic_mag_angle}, it has been shown in  \cite[Theo. $4$]{bhz3} and \cite[Theo. $5$]{bhz4} that the Chern number of the Bloch bundle associated with the flat bands at energy $\pm m$  is $\mp 1$ or more generally (including $m=0$) for the Hamiltonian in \eqref{eq:Hamiltonian}
\begin{equation}
\label{eq:Chern}
\operatorname{Cher}(P_{\ker(D(\alpha))})=-1 \text{ and } \operatorname{Cher}(P_{\ker(D(\alpha)^*)})=1.
\end{equation}
The Chern number can be computed from the expression for the Hall conductivity $\Omega(P)$, see \eqref{eq:Chern_number}, by using that 
\[ \operatorname{Cher}(P) = -2\pi i \Omega(P).\]
In particular, the net Chern number of the flat bands is zero
\[ \operatorname{Cher}(P_{\ker(D(\alpha))}\oplus P_{\ker(D(\alpha)^*)}) = \operatorname{Cher}(P_{\ker(H(0,\alpha))})=0.\]

\begin{figure}
\includegraphics[width=\textwidth]{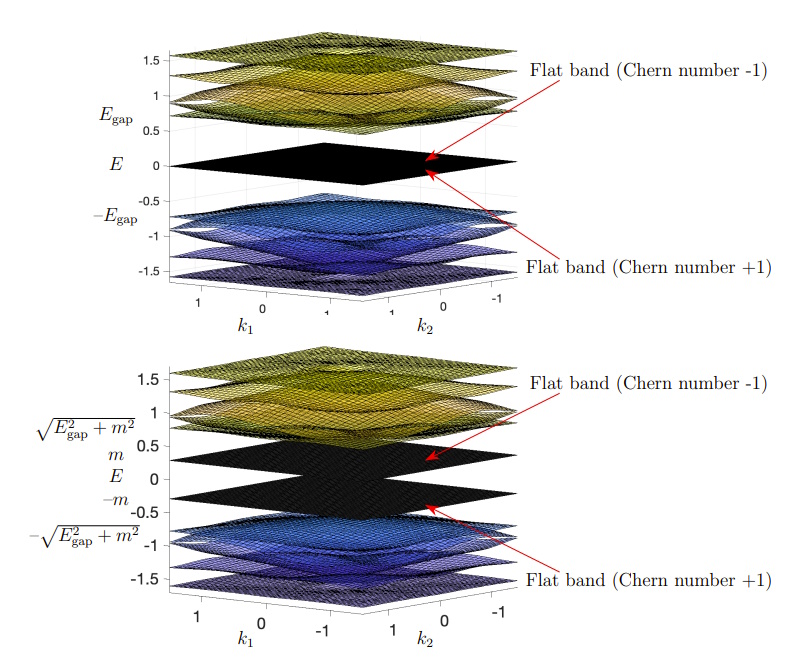}
\caption{\label{fig:gap} Band structure of non-disordered twisted bilayer graphene \eqref{eq:Hamiltonian} at the first real positive magic angle $\alpha \approx 0.58566$ with zero effective mass (top) and non-zero effective mass (bottom).}
\end{figure}

We are now going to state our assumptions on the admissible disorder profile that we consider to study Question \ref{que1}.

\begin{assumption}[Anderson model]
\label{ass:Anderson}
We consider the Anderson-type Hamiltonian with \emph{alloy-type} potentials and (possible) lattice relaxation effects with coupling strength $\lambda>0$ and $u \in C^{\infty}_{c}(\CC;\CC^4)$ of the form
\begin{equation}
\label{eq:Anderson_model}
H_{\lambda} = H + \lambda V_{X} \text{ where } V_{X} = \sum_{\gamma \in \Gamma}X_{\gamma} u(\bullet-\gamma-\xi_{\gamma}).
\end{equation}
\emph{Constraints on $X,\xi$:}
Coefficients $(X_{\gamma})_{\gamma}$ and $(\xi_{\gamma})_{\gamma}$ are families of i.i.d. random variables. 
The $(X_{\gamma})$ are assumed to be distributed according to an absolutely continuous bounded density $g$ with $\supp(g) \subset [-1,1]$. The probability measure of $(\xi_{\gamma})_{\gamma}$ is concentrated in a compact domain $D \subset \CC$. 

\emph{Constraints on $u$:} We shall impose either of the following two conditions on our matrix-valued disorder potential $u=u^* \in C^{\infty}_c(\CC;\CC^4)$.  
\begin{enumerate}
\item \emph{Case $1$:} The hermitian disorder potential $u$ is of the form
\begin{equation}
\label{eq:subtle}
 u(z) = \begin{pmatrix} Y(z) & Z(z)^* \\ Z(z) & -Y(z) \end{pmatrix} \in C^{\infty}_{c}(\CC;\CC^4)
 \end{equation}
where $ \inf_{\xi \in D^{\Gamma}} \inf_{z \in \CC} \sum_{\gamma \in \Gamma} Y(z-\gamma-\xi_{\gamma}) >0.$
\item \emph{Case $2$:} The disorder is signed, i.e. $u \ge 0$, and non-vanishing: There exist $z_0 \in \CC$ and $\varepsilon>0$ such that
\begin{equation}
\label{eq:less_subtle}
\inf_{z \in B_{\varepsilon}(z_0),\xi \in D}u(z-\xi) >0
\end{equation} as an operator. 
\end{enumerate}
For normalization purposes, we assume that $\sup_{\xi \in D^{\Gamma}} \Vert \sum_{\gamma \in \Gamma} u(\bullet-\gamma-\xi_{\gamma}) \Vert_{\infty}\le 1$ and $\supp u \subset \Lambda_R(0)$ for some fixed $R>0$ where $\Lambda_L :=\CC/(L\Gamma)$ and $\Lambda_L(z):=\Lambda_L+z.$
\end{assumption}

Random variables $\xi_{\gamma}$ model small inhomogeneities of the moir\'e lattice due to relaxation effects. 
Let us emphasize that under assumption (1) on $u$, the matrix $u$ is neither positive nor negative definite. This non-definiteness typically presents a challenge for proving Wegner estimates, as the eigenvalues do not exhibit monotonic behavior with respect to the coupling strength $\lambda$. However, we overcome this difficulty by leveraging the off-diagonal structure of the Hamiltonian. This potential-type perturbation is well-suited for studying disorder effects in the alignment of twisted bilayer graphene (TBG) with other substrates, which is essential for understanding the emergence of the anomalous quantum Hall effect (QHE). On the other hand, case (2) represents a more common scenario in the study of random Schrödinger operators, as it simplifies the proof of Wegner estimates by directly enforcing monotonicity through the positivity of the potential perturbation.

The probability space is the Polish space $\Omega = (\supp g)^{\Gamma} \times D^{\Gamma}$ with the product measure. Then $(H_{\lambda})$ is an ergodic (with respect to lattice translations) family of self-adjoint operators with continuous dependence $\Omega \ni (X,\xi) \mapsto (H_{\lambda}+i)^{-1}.$ 
Thus, there is $\Sigma \subset \mathbb R$ closed such that
\begin{equation}
\label{eq:Kirsch_inclusion}
\Spec_{L^2(\CC)}(H_{\lambda}) = \Sigma \text{ almost surely},
\end{equation} see \cite{KS80,KM82,Pa}. In addition, using ergodicity arguments, see e.g.\ \cite{W}, the density of states measure for the random operator, $\rho^{H_{\lambda}}$, exists almost surely and is almost surely non-random. In other words, $\rho^{H_{\lambda}}$ is almost surely equal to a non-random measure $\rho.$ 
An extension of our work to unbounded disorder is possible. In the context of Schr\"odinger operators this extension has been demonstrated for magnetic Landau Hamiltonians \cite{GKS2,GKM}. Furthermore, related proofs of localization for Dirac operators have also been obtained, assuming a spectral gap, in \cite{BCZ19}.

For $\lambda \neq 0$, the infinitely-degenerate point spectrum of $H$ at zero energy, corresponding to the flat band, is non-trivially perturbed and expands in energy. To capture this, we then introduce constants $K_{\pm} := \sqrt{E_{\operatorname{gap}}(\alpha)^2+m^2} \pm \lambda \sup_{X \in \Omega}  \Vert V_X \Vert_{\infty}$ and $k_{\pm}:= m \pm \vert \lambda \vert \sup_{X \in \Omega}  \Vert V_X \Vert_{\infty}.$
One thus finds analogously to \eqref{eq:spectral_gap} for the disordered Hamiltonian
\begin{equation}
\label{eq:random_spec}
 (-K_-,-k_+) \cup (-k_-,k_-) \cup (k_+,K_-) \subset \mathbb R \setminus \Sigma,
 \end{equation}
where all three intervals are non-trivial for $\lambda>0$ sufficiently small and $m >0$.
We then also define 
\begin{equation}
\label{eq:ivals}
J_-:=[-k_+,-k_-]\text{ and }J_+:=[k_-,k_+].
\end{equation}

{\begin{rem}
The condition $\lambda>0$ sufficiently small appears frequently in this text. At a fixed generic \ref{def:generic_mag_angle} magic angle, it is known that the flat bands are gapped from the remaining ones. This spectral gap however depends on the specific magic angle. Since we frequently try to keep these spectral gaps for the disordered Hamiltonian, as in \eqref{eq:random_spec}, we require $\lambda>0$ sufficiently small, depending on the gap of the specific, but fixed magic angle. 
\end{rem}}


Given a finite domain $\Lambda_L:=\CC/(L\Gamma) \subset \CC$, we introduce the Hamiltonian
\[ H_{\lambda, \Lambda_L} = H_{\Lambda_L}+\lambda V_{X,\Lambda_L},\]
with periodic boundary conditions where ${V_{X,\Lambda_L}}:=\sum_{\gamma \in \tilde \Lambda_L}X_{\gamma} (u\indic_{\Lambda_L})(\bullet-\gamma - \xi_{\gamma})\text{ with }\widetilde \Lambda_L:=\Lambda_L \cap \Gamma$. In general, we shall denote by $S_{\Lambda_L}$ the restriction of an operator $S$ to the domain $\Lambda_L$ with periodic boundary conditions in the case where $S$ is a differential operator.

The presence of a flat band in the unperturbed Hamiltonian \eqref{eq:Hamiltonian} causes a jump discontinuity in the integrated density of states (IDS). However, when considering the random Hamiltonian \eqref{eq:Anderson_model}, the IDS becomes Lipschitz continuous for all $\lambda \neq 0$. 
Due to the loss of periodicity in the randomly perturbed Hamiltonian, it is standard practice to assess the regularity of the IDS to evaluate the extent of flat band destruction.
  \begin{theo}[Continuous IDS]
  \label{theo:continuity}
The Anderson Hamiltonian as in Assumption \eqref{ass:Anderson} with $m \ge 0$ and coupling constant $\lambda \in (-\varepsilon(m),\varepsilon(m)) \setminus\{0\}$ for $\varepsilon(m) >0$ sufficiently small has almost surely H\"older continuous integrated density of states (IDS) in Hausdorff distance $d_H$ under either
\begin{itemize}
\item \emph{Case 1 disorder }\eqref{eq:subtle}: For for intervals $I,I' \subset [-k_+,k_+]$, with $k_+$ as in \eqref{eq:random_spec}, with $m >0$ or
\item \emph{Case 2 disorder }\eqref{eq:less_subtle}: For bounded intervals $I,I' \Subset \RR$, $\lambda \in \RR\setminus \{0\}$ and $m \ge 0.$ If we assume in addition that $u$ is globally positive, i.e.
\begin{equation}
\label{eq:globally_pos}
 \inf_{\xi \in D^{\Gamma}} \inf_{z \in \CC} \sum_{\gamma \in \Gamma} u(z-\gamma-\xi_{\gamma}) >0,
 \end{equation}  
 then the IDS is almost surely Lipschitz continuous
\[ \vert N(I) - N(I')\vert \lesssim_{ I,I'} d_{\text{H}}(I,I').\]
In particular, the IDS is almost surely differentiable and its Radon-Nikodym derivative, the density of states (DOS), exists almost surely and is almost surely bounded.
 \end{itemize}

 \end{theo}
 The above results follow directly from the subsequent estimate on the number of eigenvalues (NE) that imply Wegner estimates \eqref{eq:Wegner}. 
 \begin{prop}[NE]
 \label{prop:Wegner}
 Under the assumptions of Theorem \ref{theo:continuity}, we find that there is $\beta \in (0,1)$ such that
\[\mathbf E \tr(\indic_{I}(H_{\lambda,\Lambda_L})) \lesssim_{\beta} \vert  I \vert^{\beta} \vert \Lambda_L \vert.\]
If in Case 2 we assume in addition that \eqref{eq:globally_pos} holds, then we may take $\beta=1$ such that
\[ \mathbf E \tr(\indic_{I}(H_{\lambda,\Lambda_L})) \lesssim \vert I \vert \vert \Lambda_L \vert. \]
\end{prop}
\subsection{Mobility edges}
In the works of Germinet--Klein \cite{GK01,GK02,GK04} dynamical measures of transport have been introduced. The dynamical localization implies a strong form of decaying eigenfunctions; see Def. \ref{defi:Sudec}. 
To measure \emph{dynamical localization/delocalization} one introduces the following Hilbert-Schmidt norm
\begin{equation}
\label{eq:HS_quantity}
  M_{\lambda}(p, \chi, t) = \left\lVert \langle \bullet \rangle^{p/2} e^{-itH_{\lambda}} \chi(H_{\lambda})\indic_{\CC/\Gamma_3} \right\rVert_2^2,
  \end{equation}
  where $\Gamma_3:=\Gamma/3$, $\langle z \rangle:=(1+\vert z \vert^2)^{1/2}$, for some non-negative $\chi \in C_c^{\infty}(\RR)$ with time average
\[ \mathcal M_{\lambda}(p,\chi,T) = \frac{1}{T} \int_0^{\infty} \mathbf E \Big(M_{\lambda}(p, \chi, t)\Big) e^{-t/T} \ dt.\]

Recall that {$\frac{1}{T} \int_0^{\infty} t^p e^{-t/T} \ dt = T^p \Gamma(p+1)$ and $\frac{1}{T} \int_0^{\infty} e^{at} e^{-t/T} \ dt = \frac{e^{a T-1}-1}{a T-1}$, for $aT\neq 1$} to see that $\mathcal M_{\lambda}(p,\chi,T)$ indicates a time-averaged power scaling of $M_{\lambda}(p, \chi, t),${ at least for polynomial scalings of \eqref{eq:HS_quantity}.} Here, $ M_{\lambda}(p, \chi, t) $ measures the spread of mass in a spectral energy window of the Hamiltonian from the origin under the free Schr\"odinger evolution.

We shall then show that the random Hamiltonian \eqref{eq:Anderson_model} exhibits diffusive behavior in the vicinity of magic angles.
\begin{theo}[Dynamical delocalization]
\label{theo:transport}
Let $\alpha_*$ be a generic magic angle as in Definition \ref{def:generic_mag_angle}.
We consider a coupling constant $\lambda \in (-\varepsilon(m,\alpha_*),\varepsilon(m,\alpha_*))$, where $\alpha \in (\alpha_*- \delta(m,\alpha_*),\alpha_*+ \delta(m,\alpha_*))$, with mass $m \ge 0$ and sufficiently small $\varepsilon(m,\alpha_*),\delta(m,\alpha_*)>0$. The random Hamiltonian $H_{\lambda}$ demonstrates diffusive behavior for $m>0$ at no less than two energies $E_{\pm}(\lambda)$ located near $\pm m$, respectively, and at no less than one energy $E(\lambda)$ for $m=0$. 
Finally, for every $\chi \in C_c^{\infty}$ that equals one in an open interval $J$ containing at least one of $E_{\pm}(\lambda)$ and $p>0$ we have for all $T>0$
\[ \mathcal M_{\lambda}(p,\chi,T) \gtrsim_{p,J} T^{\frac{p}{4}-6}.\]
\end{theo}
{ The bound in Theorem \eqref{theo:transport} is a consequence of the transport
bounds in the region of dynamical localization. Dynamical localization will be proven in Section \ref{sec:Dyn_loc}.}
In general, we do not have a very precise understanding of how close $E_{\pm}(\lambda)$ are to $\pm m$ and how large the spectral range of dynamical delocalization is. However, by choosing a suitable disorder (of fixed support, that is, normalized strength $\lambda$, but rescaled probability), we can show that the mobility edges $E_{\pm}(\lambda)$ can be located arbitrarily close to the energies of the flat bands of the unperturbed Hamiltonian. This is discussed in Theorem \ref{theo:localization}, when $\alpha \in \mathcal A$ is a generic magic angle.
\begin{rem}
Transport behavior can also be characterized by the $p$-dependence of the estimate in the previous theorem in terms of local transport exponents
\[ \beta_{\lambda}(E) = \sup_{p>0} \inf_{\substack{I \ni E \\ I \text{ open }}} \sup_{\chi \in C_c^{\infty}(I;[0,\infty))} \liminf_{T \to \infty} \frac{\log_+ \mathcal M_{\lambda}(p,\chi,T)}{p \log(T)}.\]
The region of dynamical localization is then defined as the open set
\begin{equation}
\label{eq:DL}
\Sigma^{\operatorname{DL}} :=\{E \in \RR; \beta_{\lambda}(E)=0\},
\end{equation}
whereas the region of dynamical delocalization $\Sigma^{\operatorname{DD}}$ is defined as its complement.
A mobility edge is an energy $E \in \Sigma^{\operatorname{DD}} \cap \overline{\Sigma^{\operatorname{DL}}\cap \Sigma}.$ It follows from \cite[Theo.\ $2.10$, $2.11$]{GK04} that Theorem \ref{theo:transport} implies $\beta_{\lambda}(E_{\pm}(\lambda))>1/4.$ Theorem \ref{theo:localization} then proves the existence of mobility edges for the disordered Hamiltonian. 
\end{rem}

Although Theorem \ref{theo:transport} outlines the dynamical properties of the Hamiltonian, it is also important to explore a spectral-theoretic interpretation of transport and localization. The nature of the spectrum in the dynamically localized phase is captured by the concept of SUDEC, as stated in Definition \ref{defi:Sudec}. 

However, the presence of dynamically delocalized regimes, as described above, does not necessarily imply the existence of absolutely continuous (a.c.) or singular continuous (s.c.) spectrum. In particular, at magic angles, the Hamiltonian $H_0(m,\alpha)$ is known to exhibit an (infinitely degenerate) point spectrum at energies $\pm m$. Whether such phases can arise for our disordered Hamiltonian in the vicinity of the flat bands remains an open question. We conjecture that they do not.

As we will explain in the following, see Remark \ref{rem:uniform_decay}, the point spectrum of the Hamiltonian within an energy window containing the mobility edges, if it exists, cannot be \emph{too localized}. 

This can be made precise using the concept of generalized Wannier functions \cite{CMM19,MMP, LS21} { which applies to our setting due to the existence of Combes-Thomas estimates for the spectral projection.}

\begin{defi}[Wannier basis]
\label{def:Wannier}
Let $P$ be an orthogonal projection onto $L^2(\CC)$. We say an orthonormal basis $(\psi_{\beta})_{\beta \in I} \in L^2(\CC)$ for an index set $I \subset \mathbb N$ is an \emph{$s$-localized generalized Wannier basis} for $P$ for some $s>0$ if:
\begin{itemize}
\item $\overline{\operatorname{span}}(\psi_{\beta}) = \operatorname{ran}(P).$
\item There exists a universal $M<\infty$ and a collection of localization centers $(\mu_{\beta}) \subset \CC$ such that for all $\beta \in I$
\[ \int_{\CC} \langle z-\mu_{\beta} \rangle^{2s} \vert \psi_{\beta}(z)\vert^2 d\lambda(z) \le M, \text{ with } \lambda \text{ Lebesgue measure.}\]
\end{itemize}     
\end{defi}
Then we have the random Hamiltonian $H_{\lambda}:$
\begin{theo}[Slow decay; $m>0$]
\label{theo:spectral_deloc}
Under the assumptions of Theorem \ref{theo:transport}, we define the orthogonal projection $P_{\lambda}:= \indic_{J_{\pm}}(H_{\lambda})$ on $L^2(\CC)$ with $J_{\pm}$ as in \eqref{eq:ivals} for $m>0.$
For any $\delta>0$ and for any $\lambda \in (-\varepsilon(m),\varepsilon(m))$ with $\varepsilon(m)>0$ sufficiently small and independent of $\delta>0$, $P_{\lambda}$ does not admit a $1+\delta$-localized generalized Wannier basis. 

However, the projection admits a $1-\delta$-localized generalized Wannier basis for small disorder. 
\end{theo}
In this article, we have not considered disorder that only perturbs the off-diagonal entries of the Hamiltonian \eqref{eq:Hamiltonian}, since no techniques to show Wegner estimates for such disorder are known, which are an essential ingredient of the multi-scale analysis.

Wegner estimates are however not needed to study the decay of Wannier functions and thus we shall consider such perturbations now, by looking at the Hamiltonian
\begin{equation}
\label{eq:chiral_disorder}
H_{\lambda} = \begin{pmatrix} m I_2 & (D(\alpha) + \lambda W)^*\\  D(\alpha) + \lambda W & -m I_2 \end{pmatrix}
\end{equation}
where $W  \in L^{\infty}(\CC;\CC^{2\times 2})$ is a (possibly random) potential which we assume without loss of generality to satisfy $\Vert W \Vert_{\infty}\le 1.$
The result of Theorem \ref{theo:spectral_deloc} cannot be directly extended to $m=0$, since the net Chern number of the Hamiltonian is zero. However, the square of the Hamiltonian \eqref{eq:chiral_disorder} exhibits a diagonal form 
\begin{equation}
\label{eq:chiral_disorder_square}
H_{\lambda}^2 = \operatorname{diag}( (D(\alpha) + \lambda W)^*(D(\alpha) + \lambda W)+m^2,(D(\alpha) + \lambda W)(D(\alpha) + \lambda W)^*+m^2).
\end{equation}
Thus, to capture the low-lying spectrum, we may study the projections 
\begin{equation}
\label{eq:projections}
\begin{split}
P_{+,\lambda} &:= \indic_{[0,\mu]}((D(\alpha) + \lambda W)^*(D(\alpha) + \lambda W)) \text{ and } \\
P_{-,\lambda} &:=\indic_{[0,\mu]}((D(\alpha) + \lambda W)(D(\alpha) + \lambda W)^*), 
\end{split}
\end{equation}
separately, where we dropped the $m \ge 0$, dependence as it does not affect the spectrum apart from a constant shift. We then have
\begin{theo}[Slow decay; $m\ge 0$]
\label{theo:spectral_deloc2}
Let $\mu<E_{\operatorname{gap}}(\alpha)^2/2$ with $E_{\operatorname{gap}}(\alpha)$ as in \eqref{eq:gap} and $P_{\pm,\lambda}$ be as in \eqref{eq:projections}. For any $\delta>0$ and for any $\lambda \in (-\varepsilon,\varepsilon)$ with $\varepsilon>0$ sufficiently small and independent of $\delta>0$, projection $P_{\pm,\lambda}$ does not admit a $1+\delta$-localized generalized Wannier basis. However, the projections admit a $1-\delta$-localized generalized Wannier basis for small disorder. 
\end{theo}

We make a few observations related to Theorem \ref{theo:spectral_deloc} and the notion of Wannier bases. First, these theorems imply a lower bound on the uniform decay of eigenfunctions for the random Hamiltonian. In particular, if the random Hamiltonian exhibits a pure point spectrum, then the decay cannot be \emph{too fast} in a uniform sense. This should be compared with the notion of SUDEC, see Def. \ref{defi:Sudec} which one obtains by applying multiscale analysis.
In particular, one has
\begin{rem}[Lower bound on uniform eigenfunction decay]
\label{rem:uniform_decay}
If the Hamiltonian only exhibits point spectrum in the interval $I$, for which the associated spectral projections does not admit a $1+\delta$ generalized Wannier basis, then we can choose an orthonormal basis of eigenfunctions $(\psi_{\beta})$ such that $\overline{\operatorname{span}}(\psi_{\beta}) = \operatorname{ran}(P)$ and any sequence of localization centers $\mu_{\beta}$
\[ \sup_{\beta } \int_{\CC} \langle z-\mu_{\beta} \rangle^{2+\delta} \vert \psi_{\beta}(z)\vert^2 \ dz = \infty .\] 
In this sense, Theorem \ref{theo:spectral_deloc} gives a lower-bound on the decay of eigenfunctions in case that the random Hamiltonian exhibits only pure point spectrum.  
\end{rem}

\smallsection{Outline of article}
\begin{itemize}
\item In Section \ref{sec:stability}, we focus on Question \ref{que0} and study the (in)stability of magic angles. 
\item In Section \ref{sec:IDSWeg}, we turn to Question \ref{que1} and study the regularity of the integrated density of states by stating an estimate on the number of eigenvalues (NE) under Assumption \ref{ass:Anderson}.
\item In Section \ref{sec:mob}, we derive the existence of a mobility edge in a neighborhood of perturbed flat bands.
\item In Section \ref{sec:spec_deloc}, we prove Theorem \ref{theo:spectral_deloc}.
\end{itemize}

\section{(In)stability of magic angles} 
\label{sec:stability}
In this section, we study Question \ref{que0} and derive (in)-stability bounds on magic angles under perturbations. 
We recall the definition of the compact Birman-Schwinger operator $ T_{k } $ \eqref{eq:magic} with $ k = ( \omega^2 k_1 -
 \omega k_2 )/\sqrt 3$, where $ ( k_1, k_2 ) \in \mathbb R^2 \setminus (3 \ZZ^2 + \{ (0,0), (-1,-1) \})$. Recall that this operator is defined as 
\[   T_{k} :=(2D_{\bar z}-k)^{-1} \begin{pmatrix} 0
  &  U(z)  \\
U(-z) & 0
\end{pmatrix} :  L^2_{0}  (  \mathbb C/\Gamma ; \mathbb C^2 ) \to 
( H^1 \cap L^2_{0} ) (  \mathbb C/\Gamma ; \mathbb C^2 ) ,  \]
where
\begin{equation*}
 L^2_{ p } (  \mathbb C/\Gamma; \mathbb C^2  ) :=
\Big\{ u \in  L^2 (  \mathbb C/\Gamma, \CC^2 ) : 
\mathscr L_{\mathbf a } u(z) = e^{ 2 \pi i ( a_1 p + a_2 p ) } u(z+\mathbf a) , \ 
a_j \in \tfrac13 \mathbb Z \Big\} , \end{equation*}
for $\mathbf a = 4 \pi i( \omega a_1 + \omega^2 a_2 ).$

For scalar functions, we also define spaces $L^2_p(\CC/\Gamma;\CC)$ where we replace the translation operator \eqref{eq:La} by its first component
\eqref{eq:La}. 
As described in \eqref{eq:magic}, $\alpha \neq 0$ is magic for the unperturbed Hamiltonian if and only if $1/\alpha \in \Spec_{L^2_{0}}(T_k) \setminus\{0\}.$ 
One can then show that $1/\alpha \in \Spec_{L^2_{0}}(T_k) \setminus\{0\}$ if and only if $1/\alpha \in \Spec_{L^2_{1}}(T_k)\setminus\{0\},$ see \cite{bhz3}.
 \begin{table}
\begin{subtable}{.5\linewidth}
  \centering
\begin{tabular}{P{1.5cm} |P{2.5cm} P{2.5cm}|}
    $p$ & $3^{-p} \sigma_p\frac{\sqrt{3}}{\pi}$ \\[1ex]
    \hline\hline
    {1} & {$2/9$} \\[1ex]
    \hline
    2 & $4/9$\\[1ex]
    \hline
    3 &$32/63$ \\[1ex]
    \hline
     4 & $40/81$ \\[1ex]
    \hline
\end{tabular}
\end{subtable}%
\begin{subtable}{.5\linewidth}
  \centering
\begin{tabular}{P{1.5cm} |P{4.5cm} P{4.5cm}|}
    $p$ & $3^{-p} \sigma_p \frac{\sqrt{3}}{\pi}$ \\[1ex] 
    \hline\hline
     5 & ${9560}/{20007}$ \\[1ex]
    \hline
    6 & $ {245120}/{527877}$ \\[1ex]
    \hline
    7 & $ {1957475168}/{4337177481}$ \\[1ex]
    \hline
    8 & $ {13316086960}/{30360242367}$ \\[1ex]
    \hline
  \end{tabular}
\end{subtable}
\caption{Traces of $T_{k}^{2n}$, $\sigma_p = \tr(T_k^{2p})$, where ${\sigma_1}$ is not absolutely summable as $T_k^2$ is not of trace-class.}
\label{table:traces}
\end{table}

 We then consider a perturbation of potentials $U(z), U(-z)$ by bounded potentials $A_{\pm}, V_{\pm} \in C^{\infty}(\CC/\Gamma)$ with $(A_+,A_-) \in L^2_0$ and $\lambda>0$, where $V_{\pm}$ satisfies the same symmetries as $U(\pm \bullet)$, respectively; cf. \eqref{eq:symmetries}. This gives us a new operator $T_{k,\lambda}$ of the operator $T_k$ in \eqref{eq:magic} characterizing the new magic angles with perturbed potentials
\begin{equation}
\label{eq:Tkdelta}
 T_{k,\lambda}=(2D_{\bar z}-k)^{-1}\begin{pmatrix}\lambda A_+(z) & U(z)+\lambda V_{+}(z) \\ U(-z)+ \lambda V_{-}(z) & \lambda A_-(z) \end{pmatrix}: L^2_{1}(\CC/\Gamma;\CC^2)\to L^2_{1}(\CC/\Gamma;\CC^2).
 \end{equation}
 

To describe the spectral (in)-stability of nonnormal operators one uses \emph{pseudospectrum}, see also the book \cite[Theo. 10.2]{ET05}.
\begin{defi}
Let $P$ be a bounded linear operator. We denote the \emph{$\varepsilon$-pseudospectrum} of $P$, for every $\varepsilon>0$, by
\begin{equation}
\label{eq:perturbed}
 \Spec_{\varepsilon}(P) := \bigcup_{K \in L(H); \Vert K \Vert \le \varepsilon}  \Spec(P+K),
 \end{equation}
with $L(H)$ the space of bounded linear operators. Equivalently, it is given by
\begin{equation}
\label{eq:pseudo_spec}
\Spec_{\varepsilon}(P) = \Spec(P) \cup \{ z \notin \Spec(P); \Vert (z-P)^{-1} \Vert >1/\varepsilon\}.
\end{equation}
\end{defi}
\subsection{Stability of magic angles}
In order to study the stability of small magic angles, characterized by the eigenvalues of $T_k$ ($\alpha$ \emph{is magic} if and only if $\alpha^{-1} \in \Spec_{L^2_{1}}(T_k)$), we start with a resolvent bound and recall the definition of the regularized determinant for compact operators such that $S^2$ is a Hilbert-Schmidt operator \cite{Simon}
\[ \det_4(1+S) := \prod_{\lambda \in \Spec(S)} (1+\lambda)e^{-\lambda+\lambda^2/2-\lambda^3/3}.\] 
By the characterization of the pseudospectrum above, see \eqref{eq:pseudo_spec}, we require estimates on the resolvent to study the spectral stability.
Direct bounds on the norm of the resolvent of $T_k$ are currently not explicitly accessible and non-trivial since $T_k$ is not normal. However, estimates on the (regularized) determinant are available, since it can be expressed in terms of traces of powers of the operator $T_k$. Such traces have been studied in \cite{bhz2}.
Our approach effectively reduces the problem of magic angle stability to the analysis of the determinant of $T_k$, which can be viewed as a generalization of Cramer's rule, as demonstrated in the following lemma. We focus on the case $k=0$ to simplify the presentation and start with the resolvent bound:
 \begin{lemm}
 \label{lemm:Lemma1}
 Let $T=T_0$ be as above, then for $\alpha \in \CC$ such that $1 \notin \Spec_{L^2_{1}}( \alpha T)$
 \[ \Vert  (1-\alpha T)^{-1} \Vert \le (1+3\vert \alpha \vert)^2 + \frac{e^{3(4\vert \alpha\vert+1)^4/4}}{\vert \det_4(1- \alpha T)\vert} .\]
 \end{lemm}
 Before stating the proof of this lemma, we state a perturbation estimate that limits by how much the eigenvalues of $T_{0,\lambda}$ can spread by using the pseudo-spectrum. This bound is illustrated in Fig. \ref{fig:bound}.
 \begin{figure}
 \includegraphics[width = \textwidth]{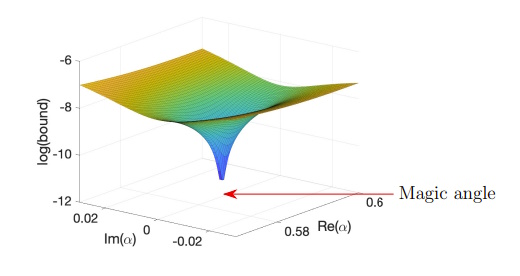}
 \caption{This figure shows the right-hand side of equation \eqref{eq:bound} close to the first magic angle.}
 \label{fig:bound}
 \end{figure}
 \begin{theo}
 \label{theo:perturbation}
 Let $T:=T_0$ and define $T_{\lambda}:=T_{0,\lambda}$ as in \eqref{eq:Tkdelta}. The perturbed operator $T_{\lambda}$ does not have any eigenvalues $\alpha^{-1}$ with $\alpha \in \CC\setminus\{0\}$ as long as the size of the perturbation satisfies
\begin{equation}
\label{eq:bound}
\Vert T_{\lambda}-T \Vert\le \frac{1}{\vert \alpha \vert \Big(  (1+3\vert \alpha \vert)^2 + \frac{e^{3(4\vert \alpha\vert+1)^4/4}}{\vert \det_4(1- \alpha T)\vert }
\Big)}.
 \end{equation}
 \end{theo}
 Before proceeding, let us discuss the meaning of \eqref{eq:bound}. By construction, we have $\Vert T_{\lambda}-T \Vert = \mathcal O(\lambda)$ with $\lambda>0$ fixed in this discussion. 
 Thus, \eqref{eq:bound} holds because of the factor $1/\vert \alpha \vert$, for small $\vert \alpha\vert$, which corresponds to large twisting angles. Thus, large angles are not magic even under random perturbations.

The right-hand side of \eqref{eq:bound} is small for large $\vert \alpha \vert$ (small twisting angles) and for $1/\alpha$ close to $\Spec_{L^2_{1}}(T)$ ($\alpha$ that are \emph{almost magic}). This means that for such $\alpha$ even small perturbations of the potential can generate eigenvalues of the form $1/\alpha$ of the perturbed operator $T_{\lambda}$. This shows that such $\alpha$ are inherently unstable, as small perturbations can create or destroy them.  In particular, this bound implies spectral stability for small $\alpha$, corresponding to large magic angles, since they remain relatively unchanged. The regularized determinant in \eqref{eq:bound} can be controlled (from above and below) by Lemma \ref{lemm:aux_lemm}.
 \begin{proof}[Proof of Theo. \ref{theo:perturbation}]
 On the one hand by the characterization of the pseudospectrum \eqref{eq:perturbed}, we find from \eqref{eq:perturbed} 
 $$\Spec_{L^2_1}(\alpha T_{\lambda}) \subset \Spec_{L^2_1, \vert \alpha \vert \Vert T_{\lambda}-T \Vert}(\alpha T).$$ This implies that if $1 \in \Spec_1(\alpha (T+R) ) ,$ with $R=\lambda (2D_{\bar z})^{-1}W$ then $1 \in \Spec_{\Vert R \Vert}(T)$. Thus, by the equivalent characterization \eqref{eq:pseudo_spec} of the pseudo-spectrum and Lemma \ref{lemm:Lemma1}
 $$\frac{1}{ \Vert \alpha R \Vert} \le \Vert (1-\alpha T)^{-1} \Vert \le (1+3\vert \alpha \vert)^2+\frac{e^{3(4 \vert \alpha\vert+1)^4/4}}{\vert \operatorname{det}_4(1-\alpha T) \vert}.$$ 
 Rearranging this estimate implies the result.
 \end{proof}

 We now give the proof of the auxiliary Lemma \ref{lemm:Lemma1}.
 \begin{proof}[Proof of Lemma \ref{lemm:Lemma1}]
 We recall from \cite[Theo.$6.4$]{Simon} that
 \begin{equation}
 \label{eq:useful_bound}
 \vert \det_4(1+S+K) \vert \le e^{ 3 \Vert S+K \Vert_4^4/4}.
 \end{equation}
Assuming $1+S$ is invertible and $S,K$ a finite rank operator, we have for the usual determinant
 \[\begin{split}
  \det(1+S+\mu K) &= \det(1+S) \det(1+\mu (1+S)^{-1} K) \\
  &= \det(1+S) (1+ \mu \tr( (1+S)^{-1} K)) + \mathcal O(\mu^2).
  \end{split}\]
This shows that 
  \[\partial_{\mu} \vert_{\mu=0}\det(1+S+\mu K)  =\det(1+S)  \tr( (1+S)^{-1} K)) \]
which shows 
   \[\partial_{\mu} \vert_{\mu=0} \log\det(1+S+\mu K) =  \tr( (1+S)^{-1} K)).\]
  Using that 
\begin{equation}
\label{eq:log_det}
 \log \det_4(1+S+\mu K) = \log \det(1+S+\mu K) - \tr(S+\mu K) + \tfrac{\tr((S+\mu K)^2)}{2}- \tfrac{\tr((S+\mu K)^3)}{3} ,
 \end{equation}
we find the log-derivative of the regularized 3-Fredholm determinant 
\[ \partial_{\mu} \vert_{\mu=0}  \log( \det_4(1+S+\mu K) ) = \tr((1+S)^{-1} K) - \tr((S^2-S+1)K).\]
By using a density argument it follows that this formula also holds for $S^2$ Hilbert-Schmidt, i.e. we can drop the assumption that $S$ is of finite rank.
Thus, from \eqref{eq:log_det} one finds specializing to $K=\langle \phi, \bullet \rangle \psi$, with $\Vert \phi \Vert = \Vert \psi \Vert =1$ and multiplying by $\det_4(1+S)$
\[ 
\det_4(1+S) \langle \phi, (1+S)^{-1} \psi \rangle = \partial_{\mu} \Big \vert_{\mu=0}\det_4(1+S+\mu K) - \det_4(1+S) \langle \phi,(S^2-S+1)\psi \rangle.
\]

Hence, using a Cauchy estimate $\vert \partial_{\mu} \vert_{\mu=0} f(\mu) \vert \le \sup_{\vert \mu \vert=1} \vert f(\mu) \vert$ for $f(\mu):= \det_4(1+S+\mu K)$, we find
\[\begin{split} \Vert \det_4(1+S) (1+S)^{-1} \Vert &\le \sup_{\vert \mu \vert=1}  \vert \det_4(1+S+\mu K) \vert+  \vert \det_4(1+S)\vert \Vert (S^2-S+1)\Vert \end{split}\]i
it thus follows together with \eqref{eq:useful_bound} that
\[\begin{split} \Vert  (1+S)^{-1} \Vert &\le \Vert S^2-S+1  \Vert + \sup_{\vert \mu \vert=1}  \frac{\vert \det_4(1+S+\mu K) \vert}{\vert \det_4(1+S)\vert}\le  \Vert S^2-S+1 \Vert + \frac{e^{3 (\Vert S \Vert_4 + 1)^4/4 }}{\vert \det_4(1+S)\vert} \end{split}.\]
Specializing the estimate to $S = - \alpha T$, we find by using that $\Vert T\Vert \le 3 $ and  $\Vert T \Vert_4\le 4$, see \cite[Lemma $4.1$]{bhz2}, that
\[ \begin{split} \Vert  (1-\alpha T)^{-1} \Vert 
&\le \Vert 1 + \alpha T+ \alpha^2 T^2 \Vert  + \frac{e^{3 (4\vert \alpha \vert  + 1)^4/4}}{\vert \det_4(1-\alpha T)\vert} \\ 
&\le 1 + 3 \vert \alpha\vert   +9\vert  \alpha^2 \vert  + \frac{e^{3 (4\vert \alpha \vert  + 1)^4/4}}{\vert \det_4(1-\alpha T)\vert} \\ 
\end{split}\]
which was to be shown.
\end{proof}
Consequently, if $\alpha_*$ is a magic angle, we can estimate $\det_4(1- \alpha T)$ in \eqref{eq:bound} by using \cite[Lemma $5.1$]{bhz3}, which in a reduced version states that
\begin{lemm}
\label{lemm:aux_lemm}
The entire function $\CC \ni \alpha \mapsto \operatorname{det}_4(1-\alpha T)$ satisfies for any $n \ge 0$
\[ \begin{split}
\left\lvert  \operatorname{det}_4(1-\alpha T) -\sum_{k=0}^{n} \mu_k \frac{(-\alpha)^k}{k!} \right\rvert &\le \sum_{j=n+1}^{\infty}   \frac{(4e^{3/4}  \vert \alpha\vert)^j}{(j!)^{1/4}} \end{split} \] 
with $\Vert A_{0} \Vert_2 \le 2$, where 
\begin{equation}
\label{eq:muk}
 {\mu}_j =  \operatorname{det}\begin{pmatrix} \sigma_1 & j-1 &0 & \cdots & 0 \\
\sigma_2 & \sigma_1 & j-2 & \cdots & 0 \\
\vdots & \vdots & \ddots & \ddots & \vdots \\ 
\sigma_{j-1}& \sigma_{j-2} & \cdots  & 0  & 1\\
\sigma_{j} & \sigma_{j-1}  & \sigma_{j-2}& \cdots & \sigma_1 \end{pmatrix}, \text{ with }  \sigma_j= \begin{cases} 
 0 & j < 4 \\
\tr T_{0}^{j} & j \ge 4. 
\end{cases}
\end{equation}
The first traces $\sigma_j$ are summarized in Table \ref{table:traces}.
\end{lemm}
\begin{proof}
This follows from the Plemelj-Smithies formula \cite[Theo. 6.8]{Simon}
and 
\cite[Theorem 7.8]{Simon} which shows that together with the above estimate $\Vert T_0\Vert_4 \le 4$
\[ \vert \mu_k \vert \le (k!)^{3/4} e^{3k/4} \Vert T_0\Vert_4^k \le (k!)^{3/4} (4e^{3/4})^k.\]
\end{proof}

\subsection{Instability of magic angles}
\begin{figure}
\includegraphics[width=\textwidth]{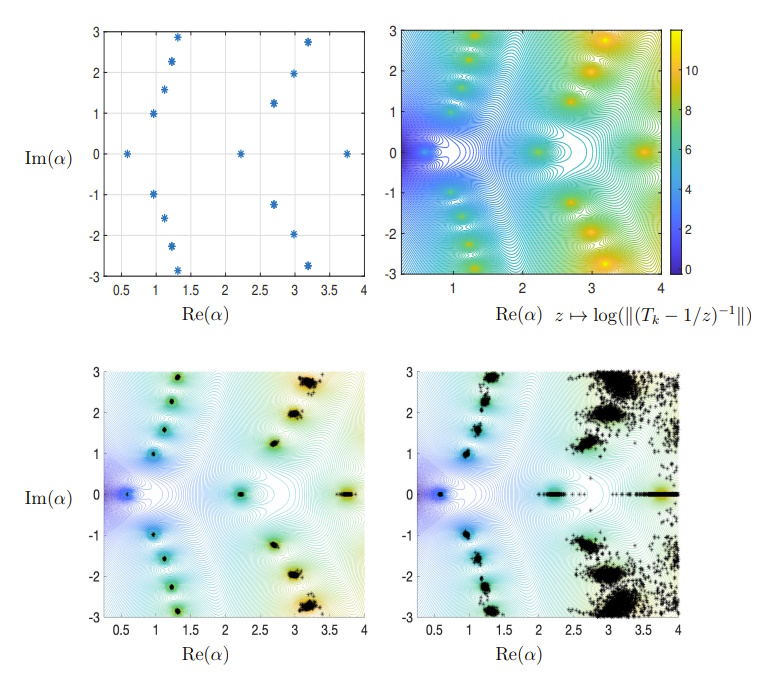}
\caption{Upper row: Magic angles (left) and resolvent norm of operator $T_k$ (right).\\
Lower row: 1000 realizations of random perturbations of tunneling potential $U+\lambda V$ with new magic angles (black dots) superimposed on resolvent norm figure. $\lambda=1/100$ (left) and $\lambda=1/10$ (right).}
\end{figure}
We shall now give the proof of Theorem \ref{theo:instability}. Arbitrary low-lying eigenvalues of $T_k$, which correspond to large magic angles in the unperturbed case, can be produced by rank $1$ perturbations of $T_k$ that are exponentially small in the spectral parameter. Let $\mu$ be one such low-lying eigenvalue of $T_k$. On the Hamiltonian side, this indicates that zero modes with quasi-momentum $k$ and $\alpha = 1/\mu$ can be generated by rank one perturbations of the Bloch-Floquet Hamiltonian, $H_k(\alpha)$

\begin{proof}[Proof of Theo. \ref{theo:instability}]
We recall that by \cite[Theo $4$]{beta} there exists for each $k \in \mathbb C$ an $L^2$-normalized $u_\mu \in C_c^{\infty}(\CC;\CC^2)$ such that the operator 
\[ P(\mu) = \begin{pmatrix} 2\mu D_{\bar z} &U(z) \\ U(-z) & 2 \mu D_{\bar z}\end{pmatrix} \] 
satisfies $\Vert (P(\mu)-\mu k)u_\mu \Vert  = \mathcal O(e^{-c/\vert \mu \vert})$ with  $\Vert u_{\mu} \Vert_{L^2} = 1$ and $c>0.$
This implies that there is a constant $K>0$, which we allow to change throughout this proof, such that $\Vert (P(\mu)-\mu k)^{-1}\Vert  \ge K e^{c/\vert \mu \vert}.$
Hence, we define the normalized $v_{\mu}:=\frac{ (P(\mu)-\mu k)u_{\mu} }{\Vert  (P(\mu)-\mu k)u_{\mu}  \Vert},$ then $\Vert (P(\mu)-\mu k)^{-1}v_{\mu}\Vert  > K e^{c/\vert \mu \vert}.$ We recall that 
\[ (P(\mu)-\mu k)^{-1} = -(T_{k}-\mu)^{-1} (2D_{\bar z}- k)^{-1}.\]
This implies that, since $\Vert (2D_{\bar z}- k)^{-1} \Vert = 1/d(k,\Gamma^*)$, where $d$ denotes the Hausdorff distance
\[\left\lVert (T_{k}-\mu)^{-1} \right\rVert \ge K e^{c/\vert \mu \vert}.\]
Hence, for the normalized $s_\mu:=\frac{(2D_{\bar z}- k)^{-1}v_\mu}{\Vert (2D_{\bar z}- k)^{-1}v_\mu\Vert}$, we have 
$$ (T_{k}-\mu)^{-1} s_\mu  =t_\mu\text{ with }\Vert t_\mu\Vert \ge K e^{c/\vert \mu\vert}.$$
Thus, we can define $R\varphi:=\frac{\langle \varphi,t_{\mu}\rangle}{\Vert t_{\mu}\Vert^2} s_{\mu}$ with norm $\Vert R \Vert =  \mathcal O(e^{-c/\vert \mu \vert})$
such that $$\mu \in \Spec(T_{k} - R).$$
\end{proof}

\section{Integrated DOS and Wegner estimate}
\label{sec:IDSWeg}

This is the first section on Question \ref{que1}. In this section, we prove Theorem \ref{theo:continuity} by providing a proof of Prop.\ref{prop:Wegner} on the regularity of the integrated density of states and prove a corresponding estimate on the number of eigenvalues of the disordered Hamiltonian. This also implies a Wegner estimate by \eqref{eq:Wegner}.
 We start with the proof of H\"older continuity using the spectral shift function, see \cite{CHK,CHK03,CHN01}, and then subsequently explain the modifications to obtain Lipschitz continuity, which uses spectral averaging.
 In the following, we will write $\chi_{x,L}:=\indic_{\Lambda_L}(x)$ with $\chi_x:=\chi_{x,1}$, $\Lambda_L:=\CC/(L\Gamma)$, and $\Lambda_L(z):=z+\Lambda_L$. Here, with some abuse of notation, we identify $\Lambda_L$ with a subset of $\CC$ centred at zero. We shall often drop subscripts to simplify the notation. For a compact operator $A$, we denote by $\Vert A \Vert_k$ the $k$-th Schatten class norm. 

\subsection{Proof of Prop. \ref{prop:Wegner}}
In this subsection we shall give the proof of Prop. \ref{prop:Wegner}, up to two crucial estimates that are provided in different subsections, namely the H\"older estimate \eqref{eq:expectation} in Subsection \ref{sec:Hoelder} and the Lipschitz estimate \eqref{eq:Lipschitz} in Subsection \ref{subsec:Lipschitz}. 

 \begin{proof}[Proof of Prop. \ref{prop:Wegner}]
 In the proof, we shall focus on Case 1 disorder, as in Assumption \ref{ass:Anderson}, as Case 2 disorder follows along the same lines. We focus on Case 1 as it requires more care, since the potential $u$ is not positive. However, we shall outline the differences of the two cases in our proof. 
 Since the spectrum in Case 1 exhibits a spectral gap, see \eqref{eq:random_spec}, we may focus without loss of generality on the spectrum around $m$. The argument around $-m$ is analogous.  In Case 2, we do not have to restrict ourselves to those neighborhoods.
Let $E_0 \in \Delta \subset \tilde \Delta \subset(k_-,k_+)$ for two closed bounded intervals $\Delta,\tilde \Delta$, with $\Delta$ of non-empty interior centered at $E_0$, and $d_0 :=d(E_0,\RR \setminus \tilde \Delta)$.
We decompose
\begin{equation}
\label{eq:disjoint}
\tr(\indic_{\Delta}(H_{\lambda,\Lambda_L}) ) = \tr(\indic_{\Delta}(H_{\lambda,\Lambda_L}) \indic_{\tilde \Delta}(H_{0,\Lambda_L}) ) + \tr(\indic_{\Delta}(H_{\lambda,\Lambda_L}) \indic_{\RR \setminus \tilde \Delta}(H_{0,\Lambda_L}) ).
\end{equation}

We then write for the second term in \eqref{eq:disjoint}
\begin{equation}
\label{eq:equa}
\begin{split}  \tr(\indic_{\Delta}(H_{\lambda,\Lambda_L}) \indic_{\RR \setminus \tilde \Delta}(H_{0,\Lambda_L}) ) &=  \tr(\indic_{\Delta}(H_{\lambda,\Lambda_L})(H_{\lambda,\Lambda_L}-E_0)(H_{0,\Lambda_L}-E_0)^{-1} \indic_{\RR \setminus \tilde \Delta}(H_{0,\Lambda_L}) )\\
&  - \tr(\indic_{\Delta}(H_{\lambda,\Lambda_L})\lambda V_{X,\Lambda_L}(H_{0,\Lambda_L}-E_0)^{-1} \indic_{\RR \setminus \tilde \Delta}(H_{0,\Lambda_L}) ).
\end{split}
\end{equation}
The first term in \eqref{eq:equa} satisfies by H\"older's inequality and the definition of $d_0$, showing $\Vert (H_{0,\Lambda_L}-E_0)^{-1} \indic_{\RR \setminus \tilde \Delta}(H_{0,\Lambda_L}) \Vert \le d_0^{-1},$
\[  \vert \tr(\indic_{\Delta}(H_{\lambda,\Lambda_L})(H_{\lambda,\Lambda_L}-E_0)(H_{0,\Lambda_L}-E_0)^{-1} \indic_{\RR \setminus \tilde \Delta}(H_{0,\Lambda_L}) ) \vert \le \frac{\vert \Delta \vert}{2 d_0} \tr(\indic_{\Delta}(H_{\lambda,\Lambda_L})).\]
We then use the inequality
\[ \begin{split}
& \tr(\indic_{\Delta}(H_{\lambda,\Lambda_L})\lambda V_{X,\Lambda_L}(H_{0,\Lambda_L}-E_0)^{-1} \indic_{\RR \setminus \tilde \Delta}(H_{0,\Lambda_L}) ) \\
&\le \frac{\vert \lambda \vert}{d_0} \Vert \indic_{\Delta}(H_{\lambda,\Lambda_L}) V_{X,\Lambda_L}\Vert_2 \Vert  \indic_{\RR \setminus \tilde \Delta}(H_{0,\Lambda_L}) \indic_{\Delta}(H_{\lambda,\Lambda_L}) \Vert_2  \\
& \le  \frac{\zeta\tr(  \indic_{\Delta}(H_{\lambda,\Lambda_L})\indic_{\RR \setminus \tilde \Delta}(H_{0,\Lambda_L}) )}{2d_0} + \frac{\lambda^2 \tr(  \indic_{\Delta}(H_{\lambda,\Lambda_L})V_{X,\Lambda_L}^2 )}{2\zeta d_0},
\end{split}\]
with $\zeta>0$. We can then bound \eqref{eq:equa}, in terms of the truncated potential
\[\tilde V_{X,\Lambda_L} := \sum_{\gamma \in \tilde\Lambda_L} u(\bullet-\gamma-\xi_{\gamma}),\] 
by choosing $\zeta>0$ sufficiently small 
\begin{equation}
\begin{split}
\label{eq:first_estm}
\tr(\indic_{\Delta}(H_{\lambda,\Lambda_L}) \indic_{\RR \setminus \tilde \Delta}(H_{0,\Lambda_L}) ) 
&\le \frac{\vert \Delta \vert}{d_0} \tr(\indic_{\Delta}(H_{\lambda,\Lambda_L}) ) +  \frac{\lambda^2 \tr(  \indic_{\Delta}(H_{\lambda,\Lambda_L})V_{X,\Lambda_L}^2 )}{\zeta d_0} \\ 
&\lesssim \frac{\vert \Delta \vert}{d_0} \tr(\indic_{\Delta}(H_{\lambda,\Lambda_L}) ) +  \frac{\lambda^2 \tr(  \indic_{\Delta}(H_{\lambda,\Lambda_L})\tilde V_{X,\Lambda_L} )}{\zeta d_0}. 
\end{split}
\end{equation}
Notice that while we do not have that $V_{X,\Lambda_L}^2 \lesssim \tilde V_{X,\Lambda_L} $, at least for Case 1 disorder, since $\tilde V_{X,\Lambda_L} $ is not positive, we still have that 
\begin{equation}
\label{eq:estm_fail} 
\tr(  \indic_{\Delta}(H_{\lambda,\Lambda_L})V_{X,\Lambda_L}^2 ) \lesssim  \tr(  \indic_{\Delta}(H_{\lambda,\Lambda_L})\tilde V_{X,\Lambda_L} ).
\end{equation}
We shall now argue this bound using a few intermediate steps. {We first recall that the spectral projection of the Hamiltonian satisfies}
\begin{equation}
\label{eq:projection}
\indic_{\Delta}(H_{0,\Lambda_L}) = P_{\ker(D(\alpha)_{\Lambda_L})} \oplus 0_{\CC^{2 \times 2}}.
\end{equation}
{This follows immediately from the structure of the Hamiltonian \eqref{eq:Bloch_Hamiltonian}. Indeed, let $(x,0)$ with $x \in \ker(D(\alpha)_{\Lambda_L})$, then it follows that 
\[H_{0,\Lambda_L}(x,0)^{\top}=m(x,0)^{\top}.\]  }

Moreover, since $\Lambda_L=\CC/(L\Gamma)$, we have by periodicity of the Hamiltonian $H_0$ that $\Spec(H_{0,\Lambda_L}) \subset \Spec(H_0).$ { This is a direct consequence of Bloch-Floquet theory.}  Since the spectrum of $H_{\lambda,\Lambda_L}${ is uniformly gapped for $\lambda \in [0,\lambda_0]$ small, as detailed in \eqref{eq:random_spec}}, it follows that the spectral projection $[0,\lambda_0]\ni \lambda \mapsto \indic_{\Delta}(H_{\lambda,\Lambda_L})$ is norm-continuous. { This is a direct consequence of the holomorphic functional calculus.} We conclude from \eqref{eq:projection} that for $\varphi = (\varphi_1,\varphi_2)$
\begin{equation}
\label{eq:continity}
 \varphi  = \indic_{\Delta}(H_{\lambda,\Lambda_L}) \varphi \Longrightarrow \Vert \varphi_2 \Vert \le \varepsilon(\lambda) \Vert \varphi_1 \Vert
 \end{equation}
with $\varepsilon(0) = 0$ and $\lambda \mapsto \epsilon(\lambda) \ge 0$ continuous. 
To see the implication in \eqref{eq:continity}, we apply norms to the left-hand side of \eqref{eq:continity}. Then, we find by substituting
\[\indic_{\Delta}(H_{\lambda,\Lambda_L}) = \indic_{\Delta}(H_{0,\Lambda_L})+ (\indic_{\Delta}(H_{\lambda,\Lambda_L})-\indic_{\Delta}(H_{0,\Lambda_L}))\]
in \eqref{eq:continity} that, { by the continuity of the spectral projection} $ \Vert \indic_{\Delta}(H_{\lambda,\Lambda_L})-\indic_{\Delta}(H_{0,\Lambda_L})\Vert = \mathcal O(\vert \lambda \vert)$, there is $C>0$ such that 
\[ \sqrt{\Vert \varphi_1 \Vert^2 + \Vert \varphi_2 \Vert^2} = \Vert \varphi\Vert  \le \Vert P_{\ker(D(\alpha)_{\Lambda_L})}\varphi_1 \Vert + C\lambda \sqrt{\Vert \varphi_1 \Vert^2 + \Vert \varphi_2 \Vert^2}.\]
Rearranging this, we find
\[(1-C\lambda) \sqrt{\Vert \varphi_1 \Vert^2 + \Vert \varphi_2 \Vert^2}  \lesssim \Vert P_{\ker(D(\alpha)_{\Lambda_L})}\varphi_1 \Vert  \le \Vert \varphi_1 \Vert.\]
Thus, we have that
\[(1-C\lambda) \sqrt{1+\Vert \varphi_2 \Vert^2/\Vert \varphi_1 \Vert^2} \le 1\]
and thus by solving this for $\Vert \varphi_2\Vert$ we find { the right-hand side of the implication in \eqref{eq:continity}
\[ \Vert \varphi_2 \Vert \le \underbrace{\frac{\sqrt{\lambda(2C-C^2\lambda)}}{1-C\lambda}}_{=:\varepsilon(\lambda)} \Vert \varphi_1 \Vert.\]
This implies in the notation of \eqref{eq:subtle} the following lower bound on the right-hand side of \eqref{eq:estm_fail}
\[ \begin{split}
\tr(  \indic_{\Delta}(H_{\lambda,\Lambda_L})\tilde V_{X,\Lambda_L} ) &= \sum_{\gamma \in \tilde\Lambda_L} \sum_{\substack{\varphi \text{ ONB of } \\
\operatorname{ran}(\indic_{\Delta}(H_{\lambda,\Lambda_L}))}} \Big( \langle \varphi_1, Y(\bullet-\gamma-\xi_{\gamma})\varphi_1\rangle- \langle \varphi_2, Y(\bullet-\gamma-\xi_{\gamma})\varphi_2\rangle \\
 &\qquad \qquad + 2\Re( \langle \varphi_1,Z(\bullet-\gamma-\xi_{\gamma}) \varphi_2\rangle ) \Big) \\ 
&\ge  \sum_{\substack{\varphi \text{ ONB of } \\
\operatorname{ran}(\indic_{\Delta}(H_{\lambda,\Lambda_L}))}} \Big(\Vert \varphi_1\Vert^2 \inf Y-k \Vert \varphi_2\Vert^2 \sup Y -2k \Vert \varphi_1\Vert \Vert \varphi_2\Vert \sup Z\Big) \\ 
& \gtrsim \sum_{\substack{\varphi \text{ ONB of } \\
\operatorname{ran}(\indic_{\Delta}(H_{\lambda,\Lambda_L}))}}\Vert \varphi_1\Vert^2 (\inf Y- \varepsilon(\lambda)^2 k \sup Y -\varepsilon(\lambda) k \sup Z) \\
& \gtrsim \sum_{\substack{\varphi \text{ ONB of } \\
\operatorname{ran}(\indic_{\Delta}(H_{\lambda,\Lambda_L}))}} \Vert \varphi_1\Vert^2 \inf Y \text{ for }\lambda \text{ small enough},
\end{split}\]
where we use that there are maximal $k \ge 1$-overlapping $\gamma$-translates and that $\lim_{\lambda \to 0} \varepsilon(\lambda)=0.$ }We can easily obtain, along the same lines, an upper bound on the left-hand side of \eqref{eq:estm_fail}{ with $M_{ij}$ denoting the $(i,j)$ entry of a matrix $M=(M_{ij}).$
\[\begin{split} \tr(  \indic_{\Delta}(H_{\lambda,\Lambda_L})V_{X,\Lambda_L}^2 ) 
& \lesssim  \sum_{\substack{\varphi \text{ ONB of } \\
\operatorname{ran}(\indic_{\Delta}(H_{\lambda,\Lambda_L}))}} \langle \varphi_1,  (V_{X,\Lambda_L}^2)_{11} \varphi_1 \rangle +  \langle \varphi_2,  (V_{X,\Lambda_L}^2)_{22} \varphi_2 \rangle + 2\Re\langle \varphi_1,(V_{X,\Lambda_L}^2)_{12} \varphi_2 \rangle  \\
& \lesssim  \sum_{\substack{\varphi \text{ ONB of } \\
\operatorname{ran}(\indic_{\Delta}(H_{\lambda,\Lambda_L}))}} \Vert \varphi_1 \Vert^2  \sup (V_{X,\Lambda_L}^2)_{11} + \Vert \varphi_2 \Vert^2  \sup (V_{X,\Lambda_L}^2)_{22} + \Vert \varphi_1 \Vert \Vert \varphi_2 \Vert  \sup (V_{X,\Lambda_L}^2)_{12} \\
& \lesssim  \sum_{\substack{\varphi \text{ ONB of } \\
\operatorname{ran}(\indic_{\Delta}(H_{\lambda,\Lambda_L}))}} \Vert \varphi_1 \Vert^2 ( \sup (V_{X,\Lambda_L}^2)_{11} + \varepsilon(\lambda)^2  \sup (V_{X,\Lambda_L}^2)_{22} +   \varepsilon(\lambda)  \sup (V_{X,\Lambda_L}^2)_{12}) \\
& \lesssim  \sum_{\substack{\varphi \text{ ONB of } \\
\operatorname{ran}(\indic_{\Delta}(H_{\lambda,\Lambda_L}))}} \Vert \varphi_1 \Vert^2 \sup (V_{X,\Lambda_L}^2)_{11},
\end{split} \]
we see that with a constant determined by the ratio of $\sup (V_{X,\Lambda_L}^2)_{11}/\inf Y$ we have shown that \eqref{eq:estm_fail} holds.}

Finally, for the first term in \eqref{eq:disjoint}, we have using the estimate
\begin{equation}
\label{eq:2bshown}
\indic_{\tilde \Delta}(H_{0,\Lambda_L}) \lesssim \indic_{\tilde \Delta}(H_{0,\Lambda_L}) \tilde V_{X,\Lambda_L}  \indic_{\tilde \Delta}(H_{0,\Lambda_L}),
\end{equation}
that we show below, the inequalities
\begin{equation}
\label{eq:rechnung}
\begin{split}  \tr(\indic_{\Delta}&(H_{\lambda,\Lambda_L})\indic_{\tilde \Delta}(H_{0,\Lambda_L}))\lesssim \tr(\indic_{\Delta}(H_{\lambda,\Lambda_L})\indic_{\tilde \Delta}(H_{0,\Lambda_L}) \tilde V_{X,\Lambda_L}  \indic_{\tilde \Delta}(H_{0,\Lambda_L})\indic_{\Delta}(H_{\lambda,\Lambda_L}))\\
&= \tr(\indic_{\Delta}(H_{\lambda,\Lambda_L})\indic_{\tilde \Delta}(H_{0,\Lambda_L}) \tilde V_{X,\Lambda_L} (\indic- \indic_{\RR\setminus \tilde \Delta}(H_{0,\Lambda_L})))\\
&= \tr(\indic_{\Delta}(H_{\lambda,\Lambda_L}) (\indic-\indic_{\RR\setminus \tilde \Delta}(H_{0,\Lambda_L}))\tilde V_{X,\Lambda_L} -\indic_{\Delta}(H_{\lambda,\Lambda_L})\indic_{\tilde \Delta}(H_{0,\Lambda_L}) \tilde V_{X,\Lambda_L} \indic_{\RR\setminus \tilde \Delta}(H_{0,\Lambda_L}))\\
&= \tr(\indic_{\Delta}(H_{\lambda,\Lambda_L})(\tilde V_{X,\Lambda_L} -\indic_{\RR\setminus\tilde \Delta}(H_{0,\Lambda_L}) \tilde V_{X,\Lambda_L} \indic_{\RR\setminus\tilde \Delta}(H_{0,\Lambda_L})  \\
&\qquad -\indic_{\RR\setminus\tilde \Delta}(H_{0,\Lambda_L}) \tilde V_{X,\Lambda_L}  \indic_{\tilde \Delta}(H_{0,\Lambda_L})  -\indic_{\tilde \Delta}(H_{0,\Lambda_L})\tilde V_{X,\Lambda_L} \indic_{\RR\setminus\tilde \Delta}(H_{0,\Lambda_L}))).
\end{split}
\end{equation}

To verify \eqref{eq:2bshown}, we proceed as follows.
Since $H_{0,\Lambda_L}$ is an unperturbed Hamiltonian, the eigenvectors associated with the spectrum in $\tilde \Delta$ are supported on the first two entries of the wavefunction, cf. \eqref{eq:projection}. Let $\pi_1:=\operatorname{diag}(\operatorname{id}_{\CC^2},0_{\CC^2})$ be the projection onto the first two entries.

We can then define another auxiliary potential $\hat V_{\Lambda_L}(z):=\inf_{\xi \in D^{\Gamma}}\sum_{\gamma \in \tilde\Lambda_L} \pi_1 u(z-\gamma-\xi_{\gamma})\pi_1.$ Thus, one has that $0 \le \hat V_{\Lambda_L} \le \pi_1 \tilde V_{X,\Lambda_L}\pi_1$ {by the positivity assumption of Case 1 in Assumption \ref{ass:Anderson}.} The projection onto the first two components is redundant for Case $2$ disorder since $u \ge 0$ in that case. 

Thus, to show \eqref{eq:2bshown}, it suffices to argue that 
$$\indic_{\tilde \Delta}(H_{0,\Lambda_L}) \lesssim \indic_{\tilde \Delta}(H_{0,\Lambda_L}) \hat V_{\Lambda_L}  \indic_{\tilde \Delta}(H_{0,\Lambda_L}).$$

Since $H_{0}$ is a periodic Hamiltonian with respect to any lattice $L\Gamma$ it suffices by Bloch-Floquet theory to prove the estimate in the Bloch function basis of the full Hamiltonian $H_{0}.$ { This is because $\Spec(H_{0,\Lambda_L})\subset \Spec(H_0)$ for any periodic subdomain of $\Lambda_L$ of the full Hamiltonian, directly by Bloch-Floquet theory.}
Indeed, let $(v_i(k))_{i \in I(k)}$ be the Bloch functions associated with the spectral projection $\indic_{\tilde \Delta}(H_{0}),$ where $I(k)$ is the set of Bloch eigenvalues inside $\tilde \Delta$ with quasimomentum $k$, where $H_{0,\Lambda_L}$ has a finite subset (in $k$) of those as eigenvectors. 
It then suffices to show that $M(k) :=( \langle v_i(k),  \hat V_{\Lambda_L}  v_j(k) \rangle_{L^2(\CC)} )_{i,j} $ is strictly positive definite for all $k$. If not, then there is $k_0 \in \CC$ and $w(k_0):=\sum_{j} \beta_j v_j$ with $\beta_j$ not all zero, such that $M(k_0) w(k_0)=0$ and by strict positivity of$\hat V_{\Lambda_L} $ on $\operatorname{ran}(\pi_1)$, see \eqref{eq:subtle}, we find $w(k_0)\vert_{B_{\varepsilon}(z_0)} \equiv 0,$ but this implies that $w\equiv 0$ by real-analyticity of $w(k_0),$ since $H_0$ is elliptic with real-analytic coefficients, which is a contradiction. Thus $M_k$ is a strictly positive matrix and using continuity in $k$\footnote{Continuity of the Bloch eigenfunctions does in general not hold, if one enforces the Bloch boundary conditions for the Bloch functions $\psi(k-\lambda) =e^{-i \bullet \lambda}  \psi(k)$. However, the Bloch boundary conditions are irrelevant for this argument. Thus, we may choose them continuously, as any vector bundle over a compact contractible space is trivial \cite[Corr. 2.17+2.18]{notes}.} and compactness of $\CC/\Gamma^*$, we also see that $M_k >c_0 >0$ for all $k$.

For the second term in the last line of \eqref{eq:rechnung}, we observe that by the boundedness of the potential
\[\vert \tr(\indic_{\Delta}(H_{\lambda,\Lambda_L})\indic_{\RR\setminus\tilde \Delta}(H_{0,\Lambda_L}) \tilde V_{X,\Lambda_L} \indic_{\RR\setminus\tilde \Delta}(H_{0,\Lambda_L}) \indic_{\Delta}(H_{\lambda,\Lambda_L}))\vert \lesssim \tr(\indic_{\Delta}(H_{\lambda,\Lambda_L})\indic_{\RR\setminus\tilde \Delta}(H_{0,\Lambda_L}))\]
where the last term can be estimated using \eqref{eq:first_estm}. 

We shall now estimate the third and fourth term at the end of \eqref{eq:rechnung} for $\delta>0$, using Young's inequality, the Cauchy-Schwarz inequality, and that $\Vert A \Vert_2 = \Vert A^*\Vert_2$
\[ \begin{split}
&\vert \tr(\indic_{\Delta}(H_{\lambda,\Lambda_L})\indic_{\RR\setminus\tilde \Delta}(H_{0,\Lambda_L}) \tilde V_{X,\Lambda_L}  \indic_{\tilde \Delta}(H_{0,\Lambda_L}) ) \vert\\
& \le \frac{\tr(\indic_{\Delta}(H_{\lambda,\Lambda_L})\indic_{\RR\setminus\tilde \Delta}(H_{0,\Lambda_L}))}{2\delta} + \frac{\delta}{2} \Vert \indic_{\Delta}(H_{\lambda,\Lambda_L})  \indic_{\tilde \Delta}(H_{0,\Lambda_L}) \tilde V_{X,\Lambda_L} \Vert_2^2 \\
& \lesssim \frac{\tr(\indic_{\Delta}(H_{\lambda,\Lambda_L})\indic_{\RR\setminus\tilde \Delta}(H_{0,\Lambda_L}))}{2\delta} + \frac{\delta}{2} \Vert \indic_{\Delta}(H_{\lambda,\Lambda_L})  \indic_{\tilde \Delta}(H_{0,\Lambda_L})\Vert_2^2
\end{split}\]
and similarly 
\[ \begin{split}
&\vert \tr(\indic_{\Delta}(H_{\lambda,\Lambda_L}) \indic_{\tilde \Delta}(H_{0,\Lambda_L}) \tilde V_{X,\Lambda_L} \indic_{\RR\setminus\tilde \Delta}(H_{0,\Lambda_L}) ) \vert\\
& \lesssim \frac{\tr(\indic_{\Delta}(H_{\lambda,\Lambda_L})\indic_{\RR\setminus\tilde \Delta}(H_{0,\Lambda_L}))}{2\delta} + \frac{\delta}{2} \Vert \indic_{\Delta}(H_{\lambda,\Lambda_L})  \indic_{\tilde \Delta}(H_{0,\Lambda_L})\Vert_2^2.
\end{split}\]
Inserting the last two estimates into \eqref{eq:rechnung} and choosing $\delta>0$ small enough
\[  \tr(\indic_{\Delta}(H_{\lambda,\Lambda_L})\indic_{\tilde \Delta}(H_{0,\Lambda_L})) \lesssim \tr(\indic_{\Delta}(H_{\lambda,\Lambda_L})\tilde V_{X,\Lambda_L} ) +\frac{\tr(\indic_{\Delta}(H_{\lambda,\Lambda_L})\indic_{\RR\setminus\tilde \Delta}(H_{0,\Lambda_L}))}{\delta}.\]
Inserting this estimate into \eqref{eq:disjoint} yields
\[\tr(\indic_{\Delta}(H_{\lambda,\Lambda_L}) ) \lesssim \tr(\indic_{\Delta}(H_{\lambda,\Lambda_L}) \indic_{\RR \setminus \tilde \Delta}(H_{0,\Lambda_L}) ) + \tr(\indic_{\Delta}(H_{\lambda,\Lambda_L})\tilde V_{X,\Lambda_L} ) +\frac{\tr(\indic_{\Delta}(H_{\lambda,\Lambda_L})\indic_{\RR\setminus\tilde \Delta}(H_{0,\Lambda_L}))}{\delta}.\]
Thus, by choosing $\vert \Delta \vert$ sufficiently small in \eqref{eq:first_estm}
\[ \tr(\indic_{\Delta}(H_{\lambda,\Lambda_L})) \lesssim \tr(\indic_{\Delta}(H_{\lambda,\Lambda_L})\tilde V_{X,\Lambda_L} ).\]

Applying expectation values and using \eqref{eq:expectation}, which we show in the next subsection, we find for $q \in (0,1)$
\begin{equation}
\label{eq:crucial_estm}
\mathbf E \tr(\indic_{\Delta}(H_{\lambda,\Lambda_L})) \lesssim \mathbf E\tr(\indic_{\Delta}(H_{\lambda,\Lambda_L})\tilde V_{X,\Lambda_L} )   \lesssim_q \vert \Delta\vert^{q} \vert \Lambda_L \vert.
\end{equation}
This shows the result by using a partition of small intervals $ \Delta $ covering $I$.
The case of Lipschitz continuity under more restrictive assumptions is relegated to Subsection \ref{subsec:Lipschitz}.
\end{proof}

\subsection{Spectral shift function and H\"older continuity}
\label{sec:Hoelder}
To obtain the H\"older estimate, used to show \eqref{eq:crucial_estm}, we recall the definition of the \emph{spectral shift function}, first.
Let $H_0$ and $H_1$ be two self-adjoint operators such that $H_1-H_0$ is trace-class, then the spectral shift function is defined as, see \cite[Ch. $8$, Sec. $2$, Theo. $1$]{Y92}
\[ \xi(\lambda,H_1,H_0) := \frac{1}{\pi}\lim_{\varepsilon \downarrow 0} \operatorname{arg}\det (\operatorname{id}+(H_1-H_0)(H_0-\lambda-i\varepsilon)^{-1}).\]
In particular for any $p \ge 1$ one has the $L^p$ bound \cite[Theorem $2.1$]{CHN01}
\begin{equation}
\label{eq:Lp}
 \Vert \xi(\bullet, H_1,H_0) \Vert_{L^p} \le \Vert H_1-H_0\Vert_{1/p}^{1/p}
 \end{equation}
where the right-hand side is defined as the generalized Schatten norm $$\Vert T \Vert_q = \Big(\sum_{\lambda \in \Spec(T^*T)} \lambda^{q/2}\Big)^{1/q}.$$

We then start by setting $\varphi(x):=\arctan(x^n)$, with $n \in 2\mathbb N_0+1$ sufficiently large, such that $h_1-h_0$ is trace-class, with $h_0 :=\varphi(H_0)$ and $h_1 :=\varphi(H_1)$. Then, we have the \emph{Birman-Krein formula}, see \cite[Ch.\ 8, Sec.\ 11, Lemma $3$]{Y92} stating that for absolutely continuous $f$
\[ \tr(f(H_1)-f(H_0)) = \int_{\RR}  \xi(\varphi(\lambda),h_1,h_0) \ df(\lambda).\]
Let $\Delta=[a,b]$ then we start by defining 
\[s(x):=\begin{cases}
0           & x \le 0 \\
3x^2 - 2x^3 & 0 \le x \le 1 \\
1           & 1 \le x \\
\end{cases}\]
and 
\begin{equation}
\label{eq:step-func}
f_{\Delta}(t) := 1-s\Bigg(\tfrac{t-a+\tfrac{1}{2}\vert \Delta \vert }{2\vert \Delta \vert}\Bigg).
\end{equation} 
We observe that this function satisfies $\inf_{t \in [1/4,3/4]} (s'(t)) = 9/8$.

Thus, we have for $C>0$
\[ \indic_{\Delta}(H_{\lambda,\Lambda_L}) \le -C\vert \Delta \vert f_{\Delta}'(H_{\lambda,\Lambda_L})\]
which implies
\[\begin{split} \tr(\lambda \tilde V_{X,\Lambda_L}  \indic_{\Delta}(H_{\lambda,\Lambda_L})) &\le -C \vert \Delta\vert \tr( \lambda \tilde V_{X,\Lambda_L}  f_{\Delta}'(H_{\lambda,\Lambda_L})) \\
&= -C\vert \Delta\vert \sum_{\gamma \in \tilde \Lambda_L} \partial_{X_{\gamma}} \tr(f_{\Delta}(H_{\lambda,\Lambda_L})).\end{split}\]
Applying the expectation value to this inequality, we find by positivity of $g$, the density of $X_{\gamma},$ that for $\mathbf E_{\gamma}$ the expectation value with respect to all random variables $(\xi_{\gamma'})$ and all $X_{\gamma'}$ apart from $\gamma' = \gamma$ 
\begin{equation}
\begin{split}
\label{eq:many_estimates} \mathbf E\tr(\lambda & \tilde V_{X,\Lambda_L}  \indic_{\Delta}(H_{\lambda,\Lambda_L}))\le-\sum_{\gamma \in \tilde \Lambda_L} \mathbf E_{\gamma} C\vert \Delta\vert \int_0^1 g(X_{\gamma}) \partial_{X_{\gamma}}\mathbf  \tr(f_{\Delta}(H_{\lambda,\Lambda_L})) \ dX_{\gamma}\\
&\le -\sum_{\gamma \in \tilde \Lambda_L} \mathbf E_{\gamma}C\vert \Delta\vert \Vert g \Vert_{\infty}  \int_0^1   \partial_{X_{\gamma}} \mathbf \tr(f_{\Delta}(H_{\lambda,\Lambda_L}))\  dX_{\gamma}\\
& \le - C\vert \Delta\vert \Vert g \Vert_{\infty}  \sum_{\gamma  \in \tilde \Lambda_L}\mathbf E_{\gamma} \tr(f_{\Delta}(H_{\lambda,\Lambda_L}(X_{\gamma}=1))-f_{\Delta}(H_{\lambda,\Lambda_L}(X_{\gamma}=0)))\\ 
&=C\vert \Delta\vert \Vert g \Vert_{\infty} \sum_{\gamma  \in \tilde \Lambda_L} \int_{\supp(f_{\Delta})} f_{\Delta}'(t)\mathbf E_{\gamma}\xi(\varphi(t),\varphi(H_{\lambda,\Lambda_L}(X_{\gamma}=1)),\varphi(H_{\lambda,\Lambda_L}(X_{\gamma}=0))) \ dt,
\end{split}
\end{equation}
where $H_{\lambda,\Lambda_L}(X_{\gamma}=\zeta)$ is the Hamiltonian $H_{\lambda,\Lambda_L}$ with $X_{\gamma}$ replaced by the constant $\zeta$ and $\vert \supp(f_{\Delta}) \vert  = \mathcal O(\vert \Delta \vert).$
Thus, using H\"older's inequality, we find for any $\beta \in (0,1)$ with \eqref{eq:Lp} and $n$ in the $\arctan$ regularization $\varphi$ sufficiently large\footnote{using $\varphi(t)-\varphi(t_0) = \int_{t_0}^t \frac{ns^{n-1}}{1+s^{2n}} \ ds$ we can create, by choosing $n$ sufficiently large, arbitrarily large powers of the resolvent. This yields the desired trace-class condition.}
\begin{equation}
\label{eq:expectation}
\mathbf E\tr(\lambda \tilde V_{X,\Lambda_L}  \indic_{\Delta}(H_{\lambda,\Lambda_L})) \lesssim \vert \Delta\vert^{\beta} \vert \Lambda_L \vert
\end{equation}
which is the identity used to obtain \eqref{eq:crucial_estm}.
\subsection{Spectral averaging and Lipschitz continuity}
\label{subsec:Lipschitz}
We now complete the proof of Lipschitz continuity for Case 2 disorder with full support, as claimed in Theorem \ref{theo:continuity} and follow an argument developed initially by Combes and Hislop \cite[Corr. $4.2$]{CH94} for Schrödinger operators.
\begin{proof}[Proof of Theorem \ref{theo:continuity} (Lipschitz continuity)]
Let $E =\max \{\vert E_1 \vert, \vert E_2 \vert\}$ where $\Delta = [E_1,E_2],$ then
\begin{equation}
\label{eq:chain_of_inequalities}
\begin{split}
 \mathbf E(\tr(\indic_\Delta(H_{\lambda,  \Lambda_L}))) 
 &\le e^{E^2} \mathbf E(\tr(\indic_\Delta(H_{\lambda,  \Lambda_L})e^{-H^2_{\lambda,  \Lambda_L}})) \\
  &\le e^{E^2} \sum_{j \in   \tilde \Lambda_L} \Bigg(\Vert \mathbf E(\chi_j \indic_\Delta(H_{\lambda, \Lambda_L}) \chi_j) \Vert \sup_{X \in \Omega} \tr\Big(   \chi_j e^{-H^2_{\lambda, \Lambda_L}}\Big)\Bigg) \\
 &\lesssim e^{E^2}  \sum_{j \in   \tilde \Lambda_L} \Vert \mathbf E( \chi_j \indic_\Delta(H_{\lambda, \Lambda_L}) \chi_j )\Vert,
\end{split}
\end{equation}
where we used that $\sup_{X \in \Omega} \tr\Big(   \chi_j e^{-H^2_{\lambda, \Lambda_L}}\Big)$ is uniformly bounded in all parameters.
Under the assumptions of Theorem \ref{theo:continuity}, we know that $u_j$ are strictly positive on $\supp(\chi_j)$ thus also $0 \le \chi_j^2 \lesssim u_j$ which is the necessary condition \cite[(4.2)]{CH94} to apply spectral averaging which readily implies together with \eqref{eq:chain_of_inequalities} that 
\begin{equation}
\label{eq:Lipschitz}
  \mathbf E(\tr(\indic_\Delta(H_{\lambda,  \Lambda_L})))  \lesssim \vert \Delta \vert \vert \Lambda_L\vert
  \end{equation}
which is the identity \eqref{eq:crucial_estm} with $\beta=1$ for Case $2$ disorder.
\end{proof}

\section{Mobility edge}
\label{sec:mob}

To prove Theorem \ref{theo:transport}, we recall the notion of summable uniform decay of correlations (SUDEC) introduced by Germinet and Klein. see \cite{GK06}.

\begin{defi}[SUDEC]
\label{defi:Sudec}
The Hamiltonian $H_{\lambda}$ exhibits a.e. \emph{SUDEC} in an interval $J$ if its spectrum in $J$ is pure point and for every closed $I \subset J$, for $\{\varphi_{n}\}$ an orthonormal set of eigenfunctions of $H_{\lambda}$ with eigenvalues $E_{n} \in I$, we define $\beta_{n}:=\Vert \langle z \rangle^{-2} \varphi_{n} \Vert^2.$ Then for $\zeta \in (0,1)$ there is $C_{I, \zeta} < \infty$ such that
\[\Vert \chi_z (\varphi_{n} \otimes \varphi_{n}) \chi_w \Vert \le C_{I, \zeta} \beta_{n} \langle z \rangle^2\langle w \rangle^2  e^{-\vert z-w \vert^{\zeta}}\text{ for }w,z\in \CC\]
and in addition one has $\mathbf P$-almost surely
\begin{equation}
\label{eq:SUDEC}
\sum_{n \in \mathbb N} \beta_{n}  < \infty.
\end{equation}
\end{defi}

The strategy to establish delocalization is to show that if the Hamiltonian would exhibit only SUDEC-type localization (SUDEC), then this would contradict the non-vanishing Chern numbers of the flat bands.
\subsection{The ingredients to the multi-scale analysis}
For the applicability of the multi-scale analysis \`a la Germinet-Klein we require six ingredients of our Hamiltonian often referred to by acronyms in their works, see also \cite{GK01},
\begin{itemize}
\item Strong generalized eigenfunction expansion {\bf{SGEE}} (Lemma \ref{lemm:SGEE}), 
\item Simon-Lieb inequality {\bf{SLI}} and exponential decay inequality {\bf{EDI}} (both Lemma \ref{lemm:SLIEDI}), 
\item Number of eigenvalues estimate {\bf{NE}} and Wegner estimate {\bf{W}} (both \eqref{eq:Wegner} and Prop.  \ref{prop:Wegner}), and 
\item Independence at a distance {\bf{IAD}}.
\end{itemize} 
The independence at a distance (IAD) just follows from the choice of Anderson-type randomness and means that the disordered potentials at a certain distance are independent of each other.

We then start with the strong generalized eigenfunction expansion (SGEE). Therefore, we introduce Hilbert spaces
\[ \mathcal H_{\pm} := L^2(\CC, \CC^4;  \langle z \rangle^{\pm 4 \nu} \ dz).\]
\begin{lemm}[SGEE]
\label{lemm:SGEE}

Let $\nu > 1/2$. The set $D^{X}_+:=\{\phi \in D(H_{\lambda}) \cap \mathcal H_+; H_{\lambda}\phi \in \mathcal H_+\}$ is dense in $\mathcal H_+$ and a core of $H_{\lambda}.$
Moreover, for $\mu \in \mathbb R \setminus \{0\}$ we have \[
\tag{SGEE} \mathbf E\left\lvert \tr\Big(\langle z \rangle^{-2\nu}(H_{\lambda}-i\mu)^{-4}\indic_I(H_{\lambda}) \langle z \rangle^{-2\nu} \Big) \right\rvert^2 <\infty.\]
\end{lemm} 
\begin{proof}
The statement about the core is immediate, as $C_c^{\infty}(\CC;\CC^4)$ is a core, see for instance Theorem \ref{theo:essentially_self}. The second statement follows as $\langle z \rangle^{-2\nu}(H_{\lambda}-i\mu)^{-2}$ is a uniformly bounded (in $X$) Hilbert-Schmidt operator. This follows for instance from \cite[Prop.9.2]{DS10}.
\end{proof}
The next lemma covers two important concepts. The Simon-Lieb inequality (SLI) relates resolvents at different scales. The eigenfunction decay inequality (EDI) connects the decay of finite-volume resolvents to the decay of generalized eigenfunctions, leading to Anderson localization. We thus define the characteristic function of the belt 
\[ \Upsilon_L(z):=\Lambda_{L-1}(z) \setminus \Lambda_{L-3}(z)\]
and denote it by $\Xi_{\Lambda_L(z)}.$ 
For $z \in \Gamma$ and $l>4$, we define smooth cut-off functions $\tilde \chi_{\Lambda_l(z)} \in C_c^{\infty}(\CC; [0,1])$ that are equal to one on $\Lambda_{l-3}(z)$ and $0$ on $\CC \setminus \Lambda_{l-5/2}(z).$
\begin{lemm}[SLI \& EDI]
\label{lemm:SLIEDI}
Let $J$ be a compact interval. For $L, l',l'' \in 2\mathbb N$ and $x,y',y'' \in \Gamma$ with $\Lambda_{l''}(y) \subsetneq \Lambda_{l'}(y') \subsetneq \Lambda_L(x)$, then $\mathbf P$-almost surely: If $E \in J \cap (\Spec(H_{\lambda,\Lambda_L(x)}) \cap \Spec(H_{\lambda, \Lambda_{l'}(y')}))^c$ then the Simon-Lieb inequality holds
\[\tag{SLI} \begin{split} \Vert \Xi_{\Lambda_L(x)}(H_{\lambda,\Lambda_L(x)}-E)^{-1} \chi_{\Lambda_{l''}(y)}\Vert &\lesssim_J \Vert \Xi_{\Lambda_{l'}(y')}(H_{\lambda,\Lambda_{l'}(y')}-E)^{-1} \chi_{\Lambda_{l''}(y)} \Vert \\
&\times \Vert \Xi_{\Lambda_L(x)} (H_{\lambda,\Lambda_L(x)}-E)^{-1} \Xi_{\Lambda_{l'}(y')}\Vert.\end{split}\]
Moreover, for any $z \in \Gamma$ and any generalized eigenfunction $\psi$\footnote{ $\psi$ solving $(H_{\lambda}-E)\psi=0$ and growing at most polynomially} with generalized eigenvalue $E \in J \cap \Spec(H_{\lambda,\Lambda_L(x)})^c$ satisfies $\mathbf P$-almost surely the eigenfunction decay inequality
 \[\tag{EDI}
\begin{split} \Vert \chi_z \psi \Vert &\lesssim_J \Vert \Xi_{\Lambda_L(x)} (H_{\lambda,\Lambda_L(x)}-E)^{-1}\chi_z \Vert \Vert \Xi_{\Lambda_L(x)} \psi \Vert.
\end{split}\]
\end{lemm}
\begin{proof}
\begin{enumerate}
\item The proof of the SLI can be streamlined for linear differential operators with disorder of Anderson-type. We start from the following resolvent identity
\[ \begin{split}
(H_{\lambda}-E)\tilde \chi_{\Lambda_{l'}(y')}(H_{\lambda,\Lambda_L(x)}-E)^{-1} &= [ (H_{\lambda}-E),\tilde \chi_{\Lambda_{l'}(y')}](H_{\lambda,\Lambda_L(x)}-E)^{-1} \\
&\ + \tilde\chi_{\Lambda_{l'}(y')} (H_{\lambda}-E) (H_{\lambda,\Lambda_L(x)}-E)^{-1}.
\end{split}\]
Using that by assumption $\Lambda_{l'}(y') \subset \Lambda_{L}(x)$ we have $\chi_{\Lambda_{l'}(y')}H_{\lambda} = \chi_{\Lambda_{l'}(y')}H_{\lambda,\Lambda_L(x)}$ and find by substituting $\chi_{\Lambda_{l'}(y')}H_{\lambda} $ in the last line above 
\[ (H_{\lambda}-E)\tilde \chi_{\Lambda_{l'}(y')}(H_{\lambda,\Lambda_L(x)}-E)^{-1} = [ H_{\lambda},\tilde \chi_{\Lambda_{l'}(y')}](H_{\lambda,\Lambda_L(x)}-E)^{-1} + \tilde \chi_{\Lambda_{l'}(y')}.\]
Since $H_{\lambda}  \tilde\chi_{\Lambda_{l'}(y')} = H_{\lambda,\Lambda_{l'}(y')} \tilde\chi_{\Lambda_{l'}(y')}$ we find by multiplying the previous line by $(H_{\lambda,\Lambda_{l'}(y')}-E)^{-1}$ that 
\[ \begin{split} \tilde \chi_{\Lambda_{l'}(y')}(H_{\lambda,\Lambda_L(x)}-E)^{-1} &= (H_{\lambda,\Lambda_{l'}(y')}-E)^{-1}[ H_{\lambda,\Lambda_{l'}(y'),},\tilde \chi_{\Lambda_{l'}(y')}](H_{\lambda,\Lambda_L(x)}-E)^{-1} \\
&\quad + (H_{\lambda,\Lambda_{l'}(y')}-E)^{-1}\tilde\chi_{\Lambda_{l'}(y')}.
\end{split}\] 
Multiplying this equation from the left by $\chi_{\Lambda_{l''}(y)}$ and from the right by $\Xi_{\Lambda_L(x)}$, the SLI ready follows from the boundedness of $[ H_{\lambda,\Lambda_{l'}(y'),},\tilde \chi_{\Lambda_{l'}(y')}]$ and submultiplicativity of the operator norm, as $\tilde \chi_{\Lambda_{l'}(y')} \Xi_{\Lambda_L(x)}=0$ implies that the second term on the right vanishes and 
\begin{equation}
\label{eq:commutator}
[ H_{\lambda,\Lambda_{l'}(y'),},\tilde \chi_{\Lambda_{l'}(y')}] = \Xi_{\Lambda_{l'}(y')}[ H_{\lambda,\Lambda_{l'}(y'),},\tilde \chi_{\Lambda_{l'}(y')}] \Xi_{\Lambda_{l'}(y')}.
\end{equation}

\item For the proof of the EDI, it suffices to choose $\psi$ as in the Lemma and observe the resolvent identity $(H_{\lambda, x, L}-E)^{-1}[H_{\lambda}, \tilde \chi_{\Lambda_L(x)} ]\psi=\tilde \chi_{\Lambda_L(x)} \psi$ which is easily verified by using $(V_X-V_{X,\Lambda_L(x)}) \tilde \chi_{\Lambda_L(x)}=0$ and $H_{\lambda} \psi = E \psi.$
Using then an analog of \eqref{eq:commutator}, $[H_{\lambda}, \tilde \chi_{\Lambda_L(x)} ]= \Xi_{\Lambda_L(x)}[H_{\lambda}, \tilde \chi_{\Lambda_L(x)} ]\Xi_{\Lambda_L(x)},$ together with the boundedness of the commutator shows the claim.
\end{enumerate}
\end{proof}
We complete our preparations by discussing the estimate on the number of eigenvalues (NE) and the Wegner estimate (W). The estimate on the number of eigenvalues (NE) is stated in Proposition \ref{prop:Wegner}. 
The Wegner estimate is then obtained by applying the estimate in Proposition \ref{prop:Wegner} to the last expression in this set of inequalities
\begin{equation}
\label{eq:Wegner}
\begin{split} \mathbf P(d(\Spec(H_{\lambda,\Lambda_L}),E)< \eta) &= \mathbf P(\operatorname{rank}\indic_{(E-\eta,E+\eta)}(H_{\lambda,\Lambda_L}) \ge 1)\\
& \le \mathbf E(\tr(\indic_{(E-\eta,E+\eta)}(H_{\lambda,\Lambda_L}) )). \end{split}
\end{equation}
\subsection{Dynamical delocalization}
\label{sec:Dyn_loc}
In this subsection we prove Theorem \ref{theo:transport}.
To replicate the proof of delocalization in \cite{GKS}, we shall study the third power of the random Hamiltonian \eqref{eq:Anderson_model}, since $H^3(M) \hookrightarrow L^2(M)$, for $M$ a two-dimensional compact manifold, is a trace-class embedding\footnote{Recall that $\lambda_k \sim_{M} k$ is the Weyl asymptotics of the negative Laplacian in dimension 2; thus $\sum_k k^{-3/2} <\infty$} and $x \mapsto x^3$ is bijective, by defining
\[ S_{\lambda}:=H_{\lambda}^3,\]
{ where we raise $H_{\lambda}$ to the third power, as $(S_{\lambda}+i)^{-1} \indic_{\Lambda_L(x)}$ is trace-class.}
Let $\mathcal C_{\pm}:=\partial B_{\vert \lambda \vert \sup_{X \in \Omega}  \Vert V_X \Vert_{\infty}}(\pm m)$ such that $\mathcal C_{\pm}$ encircles the spectrum of the random perturbation of a single flat band, but nothing else (if $m=0$, then $\mathcal C_{\pm}$ both coincide, we shall explain the modifications of this case at the end of this section). This is possible for sufficiently small noise $\lambda>0$ as the flat bands at energies $\pm m$ are strictly gapped \eqref{eq:gap} from all other bands, in the absence of disorder.  We then define the $L^2(\CC;\mathbb C^4)$-valued spectral projection
\begin{equation}
\label{eq:spectral_proj}
 P_{\lambda,\pm} = -\frac{1}{2\pi i } \int_{\mathcal C_{\pm}^3} (S_{\lambda}-z)^{-1} \ dz,
 \end{equation}
 where by $\mathcal C_{\pm}^3$ we just mean the set of elements in $\mathcal C_{\pm}$ raised to the third power.
The delocalization argument rests on the following two pillars:
\begin{itemize}
\item If the random Hamiltonian exhibits only dynamical localization close to $\pm m$, then this implies that the partial Chern numbers of $ P_{\lambda,\pm}$, defined in section \ref{sec:Partial_Chern}, have to vanish, see Prop. \ref{prop:SUDEC}.
\item The partial Chern numbers of $ P_{\lambda,\pm}$ are invariant under disorder as well as small perturbations in $\alpha$ away from perfect magic angles.
\end{itemize}
As a consequence, the Hamiltonian exhibits dynamical delocalization at energies close to $\pm m$. To simplify the notation, we drop the $\pm$ and focus solely on $+m$, since $-m$ can be treated analogously. 

The central object in this discussion is the Hall conductance. Assuming 
\[ \Vert P[[P,\Theta_1],[P,\Theta_2]] \Vert_1<\infty \]
 for a spectral projection $P$ and multiplication operators $\Theta_1(z):=\indic_{[1/2,\infty)}(\Re z)$ and $\Theta_2(z):=\indic_{[1/2,\infty)}(\Im z),$ Hall conductance is defined by
\begin{equation}
\label{eq:Chern_number}
 \Omega(P) := \tr(P[[P,\Theta_1],[P,\Theta_2]]) = \tr([P\Theta_1 P , P \Theta_2P]).
 \end{equation}
Here, $\kappa = -i [P \Theta_1 P , P \Theta_2 P]$ is also called the \emph{adiabatic curvature} with Hall charge transport $Q = -2\pi \tr(\kappa).$ {That for projections \eqref{eq:spectral_proj} $Q$ is almost surely constant is discussed in \eqref{eq:constant}. That $Q$ is an integer is shown for example in \cite[Theorem $8.2$]{ASS} or \cite[(49),(58)]{BES94} where it is related to Chern characters and Fredholm indices, respectively. See also \cite{B88} for an interpretation of the expression \eqref{eq:Chern_number} just in terms of projections in a suitable operator algebra that does not rely on the specific underlying operator.} In \cite[Theorem $1$]{BES94}, this quantity is discussed for periodic and quasi-periodic operators. 

\begin{proof}[Proof of Theo. \ref{theo:transport}]
Since $H^3(M) \hookrightarrow L^2(M) \xrightarrow[]{\text{extension}}  L^2(\CC) $ is a trace-class embedding, for bounded open sets $M$, it follows that there is a universal constant $K_1>0$ such that for sufficiently small disorder $\lambda$  and $\mu \in \mathcal C^3$ with $\mathcal C^3$ as above in trace norm 
\begin{equation}
\label{eq:trace_norm}
  \Vert (S_{\lambda}-\mu)^{-1} \chi_z \Vert_1 \le K_1 \text{ for all } z \in \Gamma.
  \end{equation}
Next, we are going to construct an analog of the Combes-Thomas estimate (CTE) for the operator $S_{\lambda}$: 

By conjugating the operator $S_{\lambda}$ with $e^f$ where $f$ is some smooth function, we find
\[ e^f S_{\lambda} e^{-f} = S_{\lambda} + R_f,\]
where $$\Vert R_f \Vert_{L(H^3, L^2)} \lesssim \varepsilon \text{ if } \Vert \partial^{\beta} f \Vert_{\infty} \le \varepsilon \ll 1 \text{ for all } 1 \le \vert \beta \vert \le 3.$$
This implies that for $z \notin \Spec(S_{\lambda})$
\[ e^f (S_{\lambda}-z)e^{-f} = (\operatorname{id}+R_f (S_{\lambda}-z)^{-1}) (S_{\lambda}-z).\]
Thus, for $z \notin \Spec(S_{\lambda})$ and $\varepsilon>0$ sufficiently small such that $\Vert R_f (S_{\lambda}-z)^{-1}\Vert <1$, $$\Vert e^{-f} (S_{\lambda}-z)^{-1}e^{f}  \Vert_{L(L^2,H^3)}  = \mathcal O(\langle d(\Spec(S_{\lambda}),z)^{-1}\rangle).$$
We conclude that for $f(z):=\varepsilon \langle z-w_0 \rangle$ with $w_0\in \CC$ fixed, we have for all $w \in \CC$
\begin{equation}
\label{eq:CTE}
\tag{CTE}
\begin{split}
 \Vert \chi_{w_0} (S_{\lambda}-z)^{-1} \chi_w \Vert &= \Vert \chi_{w_0} e^{f} (e^{-f} (S_{\lambda}-z)^{-1}e^{f} ) e^{-f} \chi_w \Vert = \mathcal O\Big(\tfrac{e^{-\varepsilon \langle w-w_0 \rangle }}{d(\Spec(S_{\lambda}),z)}\Big),
\end{split}
\end{equation}
as well as
\begin{equation}
\begin{split}
\Vert \chi_{w_0}(S_{\lambda}-S_0)(S_{\lambda}-z)^{-1} \chi_w \Vert &=\Vert \chi_{w_0} e^{f} \Vert \Vert  e^{-f}(S_{\lambda}-S_0) e^{f} \Vert_{L(H^3,L^2)} \\
& \quad \times \Vert e^{-f} (S_{\lambda}-z)^{-1}e^{f}  \Vert_{L(L^2,H^3)} \Vert  e^{-f} \chi_w \Vert \\
&=\mathcal O( \tfrac{\Vert  e^{-f} \chi_w \Vert}{d(\Spec(S_{\lambda}),z)}) =   \mathcal O\Big(\tfrac{e^{-\varepsilon \langle w-w_0 \rangle }}{d(\Spec(S_{\lambda}),z)}\Big).
\end{split}
\end{equation}
From the Combes-Thomas estimate \eqref{eq:CTE} and 
\eqref{eq:spectral_proj} we find the exponential estimate 
\begin{equation} 
\label{eq:exp_estm}
\Vert \chi_{w_0} P_{\lambda} \chi_{w} \Vert \lesssim   e^{-\varepsilon \vert w-w_0 \vert}.
\end{equation}
By \cite[Lemma $3.1$]{GKS}, this implies that $$\Vert P_{\lambda} [[P_{\lambda},\Theta_1],[P_{\lambda},\Theta_2]]\Vert_1<\infty,$$
which implies that the Hall conductance is well-defined.
In fact, using \eqref{eq:trace_norm} we have
\begin{equation}
\label{eq:CT2} \begin{split}
\Vert \chi_w P_{\lambda} \chi_{w_0} \Vert_1 =  \mathcal O(1) \text{ and } 
\Vert \chi_w P_{\lambda} \chi_{w_0} \Vert^2_2 \le \Vert \chi_w P_{\lambda} \chi_{w_0} \Vert_1\Vert \chi_w P_{\lambda} \chi_{w_0} \Vert = \mathcal O(e^{-\varepsilon \vert w-w_0 \vert}).
\end{split}
\end{equation}

To obtain the invariance of the Chern number under small disorder, we now define 
\begin{equation}
 Q_{\lambda,\zeta} := P_{\zeta}-P_{\lambda} = \frac{\zeta-\lambda}{2\pi i } \int_{\mathcal C^3} (S_{\lambda}-z)^{-1} \frac{(S_{\zeta}-S_{\lambda})}{(\zeta-\lambda)} (S_{\zeta}-z)^{-1} \ dz.
 \end{equation}
then by \eqref{eq:CT2} we find 
\begin{equation}
\label{eq:HS_norm}
\Vert \chi_w Q_{\lambda, \zeta} \chi_{w_0}\Vert_2^2= \mathcal O(e^{-\varepsilon \vert w-w_0 \vert}).
\end{equation}

If the random potential has compact support, i.e. $H_{\lambda}$ in \eqref{eq:Anderson_model} is replaced by
\begin{equation}
\label{eq:compact_support}
H_{\lambda}(L) = H + \lambda V_X \text{ where } V_X = \sum_{\gamma \in \tilde \Lambda_L}X_{\gamma} u(\bullet-\gamma-\xi_{\gamma}),
\end{equation}
for some $L>0$, then by using a partition of unity and \eqref{eq:trace_norm}, we find $\Vert Q_{\lambda,\zeta} \Vert_1<\infty$ and consequently the traces of all commutators vanish
\begin{equation}
\label{eq:small_l}
\begin{split} \Omega(P_{\zeta})-\Omega(P_{\lambda}) =& \tr([Q_{\lambda,\zeta}\Theta_1 P_{\zeta}, P_{\zeta} \Theta_2 P_{\zeta}]+[ P_{\lambda}\Theta_1Q_{\lambda,\zeta}, P_{\zeta} \Theta_2 P_{\zeta}] \\
&+[ P_{\lambda}\Theta_1P_{\lambda},Q_{\lambda,\zeta} \Theta_2 P_{\zeta}]+[ P_{\lambda}\Theta_1P_{\lambda},P_{\lambda} \Theta_2Q_{\lambda,\zeta} ])=0.
\end{split}
\end{equation}
So the integer-valued map $\lambda \mapsto \Omega(P_{\lambda}) $ is constant for $\lambda$ small around zero, under the assumption of a compactly supported random potential in \eqref{eq:compact_support}.

It remains now to drop the compact support constraint on the random potential in \eqref{eq:compact_support}.
Let $S_{\lambda}(L)=H_{\lambda}(L)^3$, then we define
\begin{equation}
\label{eq:resolvent_identity}
Q_{\lambda,>L}:=P_{\lambda}-P_{\lambda}(L) = \frac{\lambda}{2\pi i} \int_{\mathcal C^3} (S_{\lambda}-z)^{-1}(S_{\lambda}-S_{\lambda}(L))(S_{\lambda}(L)-z)^{-1} \ dz,
\end{equation}
where $P_{\lambda}(L)$ is the corresponding spectral projection associated with $S_{\lambda}(L).$
By the Combes-Thomas estimates \eqref{eq:CTE} and the resolvent identity \eqref{eq:resolvent_identity}, we find 
\[\Vert \chi_w Q_{\lambda,>L} \chi_{w_0} \Vert \lesssim e^{-c\varepsilon ((L-R-\vert w \vert)_+ + (L-R-\vert w_0 \vert)_+ + \vert w-w_0 \vert)}\]
for some $c>0$, where we used that $S_{\lambda}-S_{\lambda}(L)$ is zero on $\Lambda_{L-R}(0).$
Thus, writing the difference of Hall conductivities yields the desired limit
\begin{equation}
\label{eq:difference}
\begin{split} \Omega(P_{\lambda})-\Omega(P_{\lambda}(L)) =& \tr(Q_{\lambda,>L}[[P_{\lambda}, \Theta_1], [P_{\zeta}, \Theta_2]]+P_{\lambda,L}[ [Q_{\lambda,>L}, \Theta_1], [P_{\lambda}, \Theta_2] ] \\
&+P_{\lambda, L}[[ P_{\lambda},\Theta_1],[Q_{\lambda,>L} ,\Theta_2] ])\to 0 \text{ as } L \to \infty.\end{split}
\end{equation}
Here, one uses the strong limit $s-\lim_{L \to \infty}Q_{\lambda,>L}=0$ to show the non-vanishing of the first term on the right-hand side in \eqref{eq:difference} and that 
\[\vert \tr(P_{\lambda,L}[ [Q_{\lambda,>L}, \Theta_1], [P_{\lambda}, \Theta_2] ])\vert \le 2\sum_{\gamma,\gamma' \in \Gamma} \Vert \chi_{\gamma}  [Q_{\lambda,>L}, \Theta_1] \chi_{\gamma'} \Vert_2 \Vert \chi_{\gamma'} [P_{\lambda}, \Theta_2] \chi_{\gamma}\Vert_2 \]
with a similar estimate for the last term in \eqref{eq:difference}. The last bound converges to zero for $L \to \infty$ by using \eqref{eq:CT2} and \eqref{eq:HS_norm}, see \cite[Lemma 3.1 (i)]{GKS} for details. 
Thus, the conductivity derived from $P_{\lambda}$ is locally constant in $\lambda$ and $\alpha$, see \eqref{eq:small_l}, which shows using \eqref{eq:Chern} that Chern numbers stay $\pm 1$, for $m >0$, respectively.

For $m=0$, we repeat the previous computation with our modified $\Omega_i$  \eqref{eq:definition} to arrive at the same conclusion. Thus, if, in the notation of \eqref{eq:random_spec}, $\Sigma \cap (-K_-,K_-) \subset \Sigma^{\operatorname{DL}}$ then this would contradict the non-vanishing of the (partial) Chern number, see \eqref{eq:P0}, in regions of full localization as shown in Prop. \ref{prop:SUDEC}.

The bound in the statement of Theorem \ref{theo:transport} follows then from \cite[Theo $2.10$]{GK04}.

\end{proof}

\subsection{Dynamical localization}
Working under assumptions \eqref{ass:Anderson}, we shall now study the localized phase of the Anderson model of the form
\begin{equation}
\label{eq:Anderson2}
 H_{\lambda} = H + V_X \text{ where } V_X = \sum_{\gamma \in \Gamma}X_{\gamma} u(\bullet-\gamma-\xi_{\gamma})
\end{equation}
{with $u$ as in Case 1. We eliminated the parameter $\lambda$ in the Hamiltonian above, because a small positive $\lambda$ could easily position the metal-insulator transition near the flat bands, which we want to avoid. Instead, we select a probability distribution with a fixed support while progressively concentrating more mass near zero. }Thus, we consider random variables $X_{\gamma}$ that are distributed according to a bounded density $g_{\lambda}$ with compact support in $[-\delta,\delta]$ with $\delta<\min(m,E_{\operatorname{gap}})$ for $m> 0$ and $\delta<E_{\operatorname{gap}}$ for $m=0$.
Here, $g_{\lambda}$ is a rescaled distribution $g_{\lambda}(u) = c_{\lambda} g(u/\lambda)/\lambda \indic_{[-\delta,\delta]}$, with $g>0$,  such that as $\lambda \downarrow 0$ the mass becomes concentrated near zero and $c_{\lambda}\le C$, uniformly in $\lambda$, is the normalization constant. 

By \eqref{eq:Kirsch_inclusion}, the spectrum $\Sigma $ is almost surely independent of $\lambda.$ Our next theorem shows that the mobility edges can be shown to be located arbitrarily close to the original flat bands, by choosing $\lambda$ small, while keeping the support of the disorder fixed, within the interval $[-\delta,\delta]$. This is the motivation for our modification of the Hamiltonian \eqref{eq:Anderson2}. 
{\begin{theo}[Mobility edge]
\label{theo:localization}
Let $\langle \bullet \rangle^{n} g$ be bounded for some $n>3$ and let $\tau \in (0,\tfrac{n-3}{n+1})$. Let $H_{\lambda}$ be as in Assumption \ref{ass:Anderson} with the modification that $\lambda$ is incorporated in the rescaled density, as described in \eqref{eq:Anderson2}  and $D \subset \mathbb C$ small enough. Then for any $m>0$ there exist at least two distinct dynamical mobility edges, denoted by $\mathscr E_{+}(\lambda) >\mathscr E_{-}(\lambda)$ such that
\[ \left\lvert \mathscr E_{+}(\lambda) - m \right\rvert+\left\lvert \mathscr E_{-}(\lambda) + m \right\rvert \lesssim \lambda^{1-\frac{4}{n+1} -\tau} \xrightarrow[ \lambda \downarrow 0]{} 0.\] 
In particular, 
\[ \left\{ E \in \Big(-\sqrt{E_{\operatorname{gap}}^2/2+m^2},\sqrt{E_{\operatorname{gap}}^2/2+m^2}\Big); \vert E \pm m \vert \gtrsim \lambda^{1-\frac{4}{n+1} -\tau} \right\} \subset \Sigma^{\operatorname{DL}},\]
where the region of dynamical localization $\Sigma^{\operatorname{DL}}$ has been defined in \eqref{eq:DL}.
In the case of $m=0$, the same result is observed, but with only at least one guaranteed mobility edge. 
\end{theo}
\begin{proof}
We start by observing that using the $L^{\infty}$ bound on $\langle \bullet \rangle^{n} g$, we have for any $\varepsilon>0$ and $X \sim g_{\lambda}$ 
\begin{equation}
\begin{split}
\label{eq:flambda} \mathbf P(\vert X \vert \ge \varepsilon) &=\int_{\delta \ge \vert x \vert \ge \varepsilon} g_{\lambda}(x) \ dx \lesssim \int_{\delta/\lambda \ge \vert x \vert \ge \varepsilon/\lambda} g(x) \ dx\lesssim  \langle \lambda/\varepsilon \rangle^{n-1}.
\end{split} 
\end{equation}
Thus, for the probability of the low-lying spectrum to be contained in a small interval $[-\varepsilon,\varepsilon]$, we find for $L_0 \gg 1$ fixed and $R>0$ such that $\supp u \subset \Lambda_R(0)$

\begin{equation}
\label{eq:estimates234}
 \begin{split}
\mathbf P&\Big( \Spec(H_{\lambda,\Lambda_{L_0}} ) \cap \Big(-\sqrt{E_{\operatorname{gap}}^2/2+m^2},\sqrt{E_{\operatorname{gap}}^2/2+m^2}\Big) \subset \pm m+ [-\varepsilon,\varepsilon]\Big) \\
\overset{\text{union bound}}{\ge} &\mathbf P(\vert X_{\gamma}\vert \le \varepsilon/2; \gamma \in \tilde \Lambda_{L_0+R}(x) )  \\
\overset{\eqref{eq:flambda}}{\ge} &(1-C(\lambda/ \varepsilon)^{n-1})^{(L_0+R)^2} \\
\overset{\text{Bernoulli}}{\ge} &1- C (\lambda/ \varepsilon)^{n-1}{(L_0+R)^2},
\end{split}
\end{equation}
for small enough $\lambda/\varepsilon${, where in the following we replace $L_0+R$ just by $L_0,$ so that our estimates are valid for $D \subset \mathbb C$ small enough.}
This probability is large, if we choose 
\begin{equation}
\label{eq:cond1}
\varepsilon =C \lambda L_0^{\frac{2}{n-1}}
\end{equation}
for $C\gg 1.$ In this case, we have that $\lambda/\varepsilon = 1/(C L_0^{2/(n-1)})$ is small by choosing $C L_0^{2/(n-1)}$ large. This is precisely what we assume in \eqref{eq:estimates234} 
The choice of $\varepsilon$ in \eqref{eq:cond1} ensures that the probability in \eqref{eq:estimates234} is close to $1$. This shows that the spectrum is with high probability close to the flat band energies.
 To prove localization, one chooses $L_0 \gg 1$ large enough, as specified in \cite[(2.16)]{GK02} and $0<\lambda \ll 1$. 
We now fix an energy $\sqrt{E_{\operatorname{gap}}^2/2+m^2} \ge \vert E \vert $ such that $\vert E \pm m \vert \ge 2 \varepsilon$ with $E \in \Sigma$ . Then $E$ is, with high probability, a distance $\varepsilon>0$ from the spectrum of the finite-size Hamiltonian $H_{\lambda,\Lambda_{L_0}}$. 

To show localization, we verify the finite-size criterion of \cite[Theorem $2.4$]{GK02}. This provides another condition in addition to \eqref{eq:cond1}. In our setting, the finite-size criterion stated in \cite[Theorem $2.4$, (2.17)]{GK02} takes the following form
\begin{equation}
\label{eq:cond2}
 \frac{C_1 L_0^{25/3}}{\lambda \varepsilon}e^{-C_2 \varepsilon L_0} <1
 \end{equation}
 for two constants $C_1,C_2>0.$
 The term $L_0^{25/3}$ is obtained from \cite[Theorem $2.4$]{GK02}  by choosing (in the notation of \cite{GK02}) $b=1, d=2$, and performing a union bound over a partition of $\Gamma_0$ and $\chi_{0,L_0/3}$ which accounts for another $L_0^3.$ The $\lambda$ in the denominator is due to the scaling of the constant in the Wegner estimate which for us is proportional to the supremum norm of the density, i.e. $\Vert g_{\lambda}\Vert_{\infty} = \mathcal O(1/\lambda)$. 
 
By \cite[Theorem $2.4$]{GK02}, one concludes localization if both \eqref{eq:cond1} and \eqref{eq:cond2} hold and \eqref{eq:estimates234} holds with large probability.
 
 Setting then $\varepsilon := C_3 \lambda L_0^{\frac{2}{n-1}}$ with $C_3>0$ sufficiently large as specified in \eqref{eq:cond1}, we find that \eqref{eq:cond2} becomes
\[ \frac{C_1 L_0^{25/3}}{C_3 \lambda^2 L_0^{\frac{2}{n-1}}}e^{-C_2 C_3 \lambda^2 L_0^{\frac{n+1}{n-1}}}<1.\]
We now also set $L_0^{\frac{2}{n-1}} = \lambda^{-\frac{4}{n+1} -\tau}$ with $\tau(n)>0$ small such that $-\frac{4}{n+1} -\tau>-1.$ This means that $L_0^{\frac{n+1}{n-1}}  = \lambda^{-2 -\tau \frac{n+1}{2}}$, which implies that for $\lambda$ small enough, \eqref{eq:cond2} also holds. 

The characterization of the localized regime then follows from \cite[Theorem $2.4$]{GK02}, the existence of a mobility edge follows together with Theorem \ref{theo:transport}.
\end{proof}}

\section{Decay of point spectrum and Wannier bases}
\label{sec:spec_deloc}
We now give the proof of Theorem \ref{theo:spectral_deloc} and \ref{theo:spectral_deloc2}. We focus primarily on the first case, explaining the modifications required for the second result at the end. 
\begin{proof}[Proof of Theo.\ref{theo:spectral_deloc} \& \ref{theo:spectral_deloc2}]
We first reduce the analysis to $\lambda=0.$ By $\lambda$-continuity of the random perturbation, the spectral projections $P_0 = \indic_{J_{\pm}}(H_0)$ and $P_{\lambda} = \indic_{J_{\pm}}(H_{\lambda})$ with $J_{\pm}$ as in \eqref{eq:ivals} 
\[ \Vert P_0 - P_{\lambda} \Vert = \mathcal O(\vert \lambda \vert)\]
by using e.g. the resolvent identity and holomorphic functional calculus and the spectral gap of the Hamiltonian. 
Thus, for $\lambda$ small enough there is an isometry \cite[Lemma $10$]{BES94} \cite[Theo.6.32]{Kato} $U$, also known as the Kato-Nagy formula \cite{Ka55,SN47}, such that $UU^* = P_0$ and $U^*U=P_{\lambda}.$ In particular $P_{0}U = U P_{\lambda}.$ It then follows that $U$ has a Schwartz kernel $K$ that is exponentially close to the identity, cf. \cite[Lemma $8.5$]{CMM19}.
By this we mean that there is $m>0$ such that
\[ \vert K(z,z')-1\vert  = \mathcal O(e^{-m \vert z-z'\vert}).\]
The Schur test for integral operators implies that $\tilde U:=\langle \bullet-z_0 \rangle U  \langle \bullet-z_0 \rangle^{-1}$ is a family of operators uniformly bounded in $z_0 \in \CC$. 
This implies that for any $\varphi \in L^2(\CC;\CC^4)$ 
$$ \langle \bullet-z_0 \rangle^{1+\delta} P_{0} U\varphi = \langle \bullet-z_0 \rangle^{1+\delta} U \langle \bullet-z_0 \rangle^{-1-\delta} \langle \bullet-z_0 \rangle^{1+\delta} P_{\lambda}\varphi.$$
Taking norms, we find, using that $\Vert \langle \bullet-z_0 \rangle^{1+\delta} P_{\lambda}\varphi\Vert < \infty$ by assumption, that
\[ \Vert \langle \bullet-z_0 \rangle^{1+\delta}  P_{0} U\varphi \Vert  < \infty.\]
This implies, by choosing for $U\varphi$ an orthonormal basis of $\overline{\operatorname{ran}}(P_0)$, i.e. $(\psi_n)$ is an orthonormal basis of $\overline{\operatorname{ran}}(P_0)$, then $\varphi_n:=U^*\psi_n$, that $P_0 $ exhibits a $(1+\delta)$-localized generalized Wannier basis. Since $P_0$ is precisely the projection onto $\ker(D(\alpha))$, we deduce that $P_0$ exhibits a non-zero Chern number, see \eqref{eq:Chern}, and therefore do not possess a $(1+\delta)$-localized Wannier basis, see \cite{LS21} which gives a contradiction.

Conversely, let $P_k(\alpha)= (2\pi i)^{-1} \oint_{\gamma} (z-H_{k}(0,\alpha))^{-1} \ dz$, where $\gamma$ is a sufficiently small circle around zero that encircles only the flat band eigenvalue but nothing else in the spectrum of $H_k(0,\alpha)$. Then $P_k(\alpha)$ is the spectral projection onto the flat band eigenfunction of $H_k$. Since $k \mapsto H_k$ is real-analytic, this implies that $k \mapsto P_k$ is real-analytic. 
 Moreover, since $H_{k-\gamma^*}(\alpha) = \tau(\gamma^*) H_k(\alpha) \tau(\gamma^*)^{-1}$ with $\tau_{\gamma}(z):=e^{i\Re( z\overline{\gamma^*})}$ with $\gamma^* \in \Gamma_3^*,$ the spectral projection satisfies the covariance relation
 \[P_{k-\gamma}(\alpha) = \tau(\gamma^*) P_k(\alpha) \tau(\gamma^*)^{-1}. \]
 It then follows from \cite[Theo. $2.4$]{MPPT} that there exists an associated Wannier basis that satisfies $\Vert \langle \bullet \rangle^{p/2} w_{\gamma} \Vert^2_{L^2(\CC)}\le C <\infty$ for $p<1$ and all $\gamma \in \Gamma$ for the unperturbed periodic problem. Reversing the argument provided in the first part of the proof, it follows that the randomly perturbed problem also exhibits a Wannier basis. 
 
To show Theorem \ref{theo:spectral_deloc2} one proceeds analogously and notices that $P_{\pm,\lambda=0}$ corresponds to the projections onto $\ker(D(\alpha))$ and $\ker(D(\alpha)^*)$, each exhibiting a nonzero Chern number. 
\end{proof}

With this result at hand, we are able to evaluate the quantity \eqref{eq:HS_quantity} for the unperturbed Hamiltonian, providing a link between the dynamical and spectral theoretic notion of (de)-localization.
\begin{prop}
\label{prop:wannier}
Let $\alpha$ be a simple magic angle, as in Def. \ref{def:generic_mag_angle}, then for all $p \ge 1$
\[ \left\lVert \langle \bullet \rangle^{p/2} e^{-itH(\alpha)}P_{\ker(H(\alpha))}\indic_{\CC/\Gamma_3} \right\rVert_2^2 = \infty, \]
while this expression is finite for $p<1.$ 
\end{prop}
\begin{proof}
 We start by observing that for an orthonormal basis $(f_n)$ of $L^2(\CC/\Gamma_3)$ and $(e_i)$ the standard basis of $\CC^4$
 \[\begin{split}
 \left\lVert \langle \bullet \rangle^{p/2} e^{-itH(\alpha)}P_{\ker(H(\alpha))}\indic_{\CC/\Gamma_3} \right\rVert_2^2 
 &= \left\lVert \langle \bullet \rangle^{p/2} P_{\ker(H(\alpha))}\indic_{\CC/\Gamma_3} \right\rVert_2^2\\
 &= \left\lVert \langle \bullet \rangle^{p/2} P_{\ker(D(\alpha)) \oplus \ker(D(\alpha)^*)}\indic_{\CC/\Gamma_3} \right\rVert_2^2\\
  &= \sum_{i=1}^4 \sum_{n \in \mathbb N}\left\lVert \langle \bullet \rangle^{p/2} P_{\ker(D(\alpha)) \oplus \ker(D(\alpha)^*)} f_n \otimes e_i \right\rVert^2\\
    &= \left\lVert \langle \bullet \rangle^{p/2} P_{\ker(D(\alpha))} \indic_{\CC/\Gamma_3} \right\rVert_2^2 + \left\lVert \langle \bullet \rangle^{p/2} P_{\ker(D(\alpha)^*)} \indic_{\CC/\Gamma_3} \right\rVert_2^2.
 \end{split}\]
 Without loss of generality, we shall focus on the first summand. 
Consider the unitary Bloch-Floquet transform $\mathcal Bu(z,k):=\sum_{\gamma \in \Gamma_3} e^{i \langle z+\gamma,k\rangle} \mathscr L_{\gamma}u(z)$, where $\mathscr L_{\gamma}$ has been defined in \eqref{eq:La}, with the convention that $\langle z,z_0\rangle:=\Re(z\bar z_0),$ and its inverse/adjoint $\mathcal C  v(z):=\int_{\CC/\Gamma_3^*} v(z,k) e^{-i\langle z,k\rangle} \ \frac{dm(k)}{\vert \CC/\Gamma_3^* \vert}$. 
We then find that
\begin{equation}
\label{eq:identity2}
 \mathscr L_{\gamma} \mathcal C  v(z):=\int_{\CC/\Gamma_3^*} v(z,k) e^{-i\langle z+\gamma,k\rangle} \ \frac{dm(k)}{\vert \CC/\Gamma_3^* \vert}  = \mathcal C(e^{-i \langle \gamma,k \rangle} v(z,k)).
 \end{equation}
Since by assumption $\ker(D(\alpha)+k) = \operatorname{span}\{\varphi(\bullet,k)\}$, we see that 
\begin{equation}
\label{eq:basis}
(e^{-i \langle \gamma,k \rangle} \varphi(z,k))_{\gamma \in \Gamma}\text{, for }\varphi(\bullet,k) \in L^2(\CC/\Gamma_3)\text{ normalized,}
\end{equation}
forms a basis of the space $\int^{\oplus}_{\CC/\Gamma_3^*} \ker(D(\alpha)+k) dk$. Indeed, orthonormality just follows from 
\begin{equation}
\label{eq:orthonormality}
\begin{split}
 \langle e^{-i \langle \gamma,k \rangle} \varphi(z,k), e^{-i \langle \gamma',k \rangle} \varphi(z,k) \rangle 
 &= \int_{\CC/\Gamma_3^*} \int_{\CC/\Gamma_3} \vert \varphi(z,k)\vert^2 e^{-i\langle \gamma-\gamma',k \rangle} \frac{dz \ dk}{\vert \CC/\Gamma_3^*\vert} \\
&= \int_{\CC/\Gamma_3^*} e^{-i\langle \gamma-\gamma',k \rangle} \frac{dk}{\vert \CC/\Gamma_3^*\vert} = \delta_{\gamma,\gamma'}
 \end{split}
 \end{equation}
and completeness from the completeness of the regular Fourier expansion, i.e. a general element in this subspace is of the form 
\[ \sum_{\gamma \in \Gamma_3^*} f(\gamma) e^{-i \langle \gamma,k \rangle} v(z,k) \text{ for } f \in \ell^2(\Gamma_3^*).\]

We then have $\mathcal B D(\alpha)\mathcal C \varphi(x,k) = (D(\alpha)+k)\varphi(x,k).$.
Recall the trivial decomposition of $L^2$ given by $L^2(\CC) = L^2(\CC/\Gamma_3) \oplus L^2(\CC \setminus (\CC/\Gamma_3)).$

We then find for the Hilbert-Schmidt norm using an orthonormal basis $(e_n)$ of $L^2(\CC/\Gamma_3)$
\[ \begin{split} 
&\left\lVert \langle \bullet \rangle^{p/2} P_{\ker(D(\alpha))}\indic_{\CC/\Gamma_3} \right\rVert_2^2=\left\lVert \langle \bullet \rangle^{p/2} P_{\ker(D(\alpha))}\indic_{\CC/\Gamma_3} \right\rVert_2^2\\
&=\sum_{n \in \ZZ} \Vert  \langle \bullet \rangle^{p/2} P_{\ker(D(\alpha))}e_n \Vert^2_{L^2(\CC)} =\sum_{n \in \ZZ} \Vert  \langle \bullet \rangle^{p/2} \mathcal C P_{\ker(D(\alpha)+k)} \mathcal B \indic_{\CC/\Gamma_3} e_{n} \Vert^2_{L^2(\CC)}. \end{split} \]
Since by assumption $P_{\ker(D(\alpha)+k)} = \varphi(\bullet,k) \otimes \varphi(\bullet,k) $ is a rank $1$ projection, we have 
\[ \begin{split}
\Vert  \langle \bullet \rangle^{p/2} \mathcal C P_{\ker(D(\alpha)+k)} \mathcal B \indic_{\CC/\Gamma_3} e_{n} \Vert^2_{L^2(\CC)} &=\Vert \langle \bullet \rangle^{p/2} \mathcal C \varphi \Vert^2_{L^2(\CC)} \vert \langle \varphi, \mathcal B( \indic_{\CC/\Gamma_3} e_{n} ) \rangle_{L^2(\CC/\Gamma_3 \times \CC/\Gamma_3^*)}\vert^2 \\
&=\Vert \langle \bullet \rangle^{p/2} \mathcal C \varphi \Vert^2_{L^2(\CC)} \vert \langle \mathcal C\varphi, \indic_{\CC/\Gamma_3} e_{n}  \rangle_{L^2(\CC)}\vert^2.\end{split} \]
This implies that 
\[\begin{split} \left\lVert \langle \bullet \rangle^{p/2}P_{\ker(D(\alpha))}\indic_{\CC/\Gamma_3} \right\rVert_2^2 &= \Vert \langle \bullet \rangle^{p/2} \mathcal C \varphi \Vert^2_{L^2(\CC)} \sum_{n\in \mathbb Z} \vert \langle \mathcal C\varphi, e_{n}  \rangle_{L^2(\CC/\Gamma_3)}\vert^2\\
 &=\Vert \langle \bullet \rangle^{p/2} \mathcal C \varphi \Vert^2_{L^2(\CC)} \Vert \mathcal C \varphi \Vert_{L^2(\CC/\Gamma_3)}. \end{split}\]
 However, a Wannier basis is obtained from $\mathcal C\varphi$ by defining $w_{\gamma}:=\mathscr L_{\gamma} \mathcal C \varphi$. Indeed, using \eqref{eq:orthonormality} functions $w_{\gamma}$ are an orthonormal basis of $\ker(D(\alpha))$ as \[ \langle w_{\gamma},w_{\gamma'} \rangle_{L^2(\CC)} = \langle \mathscr L_{\gamma}\mathcal C \varphi ,  \mathscr L_{\gamma'}\mathcal C \varphi \rangle_{L^2(\CC)} = \delta_{\gamma,\gamma'}\]
 and span $\ker(D(\alpha))$ due to \eqref{eq:identity2} and \eqref{eq:basis}.  
 Thus, we obtain since $\mathscr L_{\gamma}$ is an isometry that  $$\Vert \langle \bullet \rangle^{p/2} w_0 \Vert^2_{L^2(\CC)} = \Vert \mathscr L_{\gamma} \langle \bullet \rangle^{p/2} w_0 \Vert^2_{L^2(\CC)} = \Vert \langle \bullet + \gamma \rangle^{p/2} w_{\gamma}\Vert_{L^2(\CC)}.$$
From the non-existence of a $1$-localized Wannier basis and the existence of a $(1-\delta)$ Wannier basis, for any $\delta>0$, see for instance \cite{MPPT}, we find that $\Vert \langle \bullet \rangle^{p/2} w_0 \Vert^2_{L^2(\CC)}=\infty$ for $p\ge 1$ and is finite for $p<1.$

 \end{proof}

\begin{appendix}
\section{Essential self-adjointness}

In this appendix, we recall the essential self-adjointness of our Hamiltonian with even possibly unbounded disorder on $C_c^{\infty}(\CC)$. 
\begin{theo}
\label{theo:essentially_self}
The Hamiltonian $H_{\lambda}(\alpha)$ \eqref{eq:Anderson_model} is, under the more general assumptions, with $L^{\infty}(\RR)$-bounded density $g$ for random variables $(X_{\gamma})$ and arbitrary density $h$ is almost surely essentially self-adjoint on $C_c^{\infty}(\CC).$
\end{theo}
\begin{proof}
To see that $H_{\lambda}(\alpha)$ is essentially self-adjoint, we first observe that it is symmetric on $C_c^{\infty}(\CC).$ It thus suffices to show that for any $L^2$-normalized $\psi$
\[(H_{\lambda}(\alpha) \pm i) \psi= 0 \text{ implies }\psi\equiv 0, \]
i.e. the deficiency indices are zero.
Elliptic regularity and the assumption that $u \in L^{\infty} $ implies that $\psi \in C^{\infty}(\CC)$.
We then pick a cut-off function $\eta_n(z):=\eta(z/n)$ with $\eta\in C_c^{\infty}(\CC)$ and $\eta\vert_{B_1(0)} \equiv 1$ and find 
\[ (H_{\lambda}(\alpha) \pm i) \eta_n \psi = \begin{pmatrix} 0 & 2D_{z}\eta_n \cdot \operatorname{id}_{\mathbb C^2} \\ 2 D_{\bar z}\eta_n \cdot \operatorname{id}_{\mathbb C^2} & 0 \end{pmatrix}   \psi.\]
We conclude that 
\[ \Vert \eta_n \psi \Vert_2^2 + \Vert H_{\lambda}(\alpha) \eta_n \psi \Vert_2^2 = \Vert (H_{\lambda}(\alpha)  \pm i ) \eta_n \psi \Vert_2^2 \lesssim \Vert \nabla \eta_n \Vert_{\infty}^2 = \mathcal O(1/n^2) \xrightarrow[n \to \infty]{} 0.\]
Since $\eta_n \psi \to \psi$ by dominated convergence, we conclude that $\psi \equiv 0.$
\end{proof}

\section{Partial Chern numbers \& Euler numbers}
\label{sec:Partial_Chern}
Let $P$ be an orthogonal projection on $L^2(\CC;\CC^{2n})$ such that for some $\xi \in (0,1)$, $\kappa>0$, and $K_P<\infty$ we have 

\begin{equation}
\label{eq:decay} \Vert \chi_{z_0} P \chi_{z_1} \Vert_2 \le K_P \langle z_0 \rangle^{\kappa} \langle z_1 \rangle^{\kappa} e^{-\vert z_0-z_1\vert^{\xi}} \text{ for all }z_0,z_1 \in \Gamma.
\end{equation}
Spectral projections of Hamiltonians exhibiting (SUDEC) satisfies this property, as follows directly from Def. \ref{defi:Sudec}.
Let $\pi_1:=\operatorname{diag}(\operatorname{id}_{\CC^n},0)$ and $\pi_2:=\operatorname{diag}(0,\operatorname{id}_{\CC^n})$. By the definition of the Hilbert-Schmidt norm one finds for all $i,j$
\begin{equation}
\label{eq:HS_Bound} \Vert \chi_{z_0} \pi_i P\pi_j \chi_{z_1} \Vert_2 \le \Vert \chi_{z_0} P \chi_{z_1} \Vert_2  \le K_P \langle z_0 \rangle^{\kappa} \langle z_1 \rangle^{\kappa} e^{-\vert z_0-z_1\vert^{\xi}} \text{ for all }z_0,z_1 \in \Gamma.
\end{equation}
We define the new $\hat \Theta_{j}(i) :=\pi_i \Theta_j = \Theta_j \pi_i$ and replace \eqref{eq:Chern_number} by
\begin{equation}
\label{eq:definition}
 \Omega_i(P):=\tr(P [[P,\hat \Theta_1(i)],[P,\hat \Theta_2(i)]]) 
 \end{equation}
under the assumption of
\begin{equation}
\label{eq:well-posed}
 \vert \Omega_i(P)\vert :=\Vert P [[P,\hat \Theta_1(i)],[P,\hat \Theta_2(i)]]\Vert_1 <\infty.
 \end{equation}
\begin{rem} 
It is convenient to modify $\hat \Theta_i$ rather than $P$ in the definition of $\Omega$, since $\pi_i P \pi_j$ is in general no longer a projection, even for $i=j.$
\end{rem}

Since we still have that $[\hat \Theta_i,\hat \Theta_j ]=0$ we find the equivalent formulation of \eqref{eq:definition}
\begin{equation}
\label{eq:commutator2}
 \Omega_i(P) = \tr( [P\hat \Theta_1(i) P,P\hat \Theta_2(i) P]).
 \end{equation}
In particular, if $P$ is a finite-rank projection, we always find $\Omega_i(P)=0$, as \eqref{eq:commutator2} is a commutator of trace-class operators.

To provide further motivation for the above definition \eqref{eq:definition}, we shall consider the unperturbed Hamiltonian $H_{0} = \begin{pmatrix} m & D^* \\ D & -m \end{pmatrix}$ then $H_0^2 = \operatorname{diag}(D^*D+m^2 , DD^* + m^2 ) $ and consequently any spectral projection of $H_0^2$ is also diagonal and thus of the form $P_0 = \operatorname{diag}(P_0(1) , P_0(2)).$ Thus, we have 
\[ \Omega_i(P_0) = \tr( [P_0 \hat \Theta_1(i) P_0 , P_0 \hat \Theta_2(i) P_0]) = \Omega(P_0(i)),\]
where we recall from \eqref{eq:Chern} that for a generic magic angle and $P_0 = \indic_{[0,\mu]}(H_0^2)$ with $\mu \in (0,E_{\operatorname{gap}}^2)$
\begin{equation}
\label{eq:P0}
 \Omega_1(P_0) =\frac{i}{2\pi } \text{ and } \Omega_2(P_0)=-\frac{i}{2\pi }.
 \end{equation}

Thus, while $\Omega(P_0)=0$ for $m=0$, we have $\Omega_1(P_0),\Omega_2(P_0) \neq 0.$ The definition of $\Omega_i$ captures the non-trivial sublattice Chern numbers of twisted bilayer graphene while the total Chern number vanishes. The existence of these non-zero Chern numbers is due to the PT or $C_{2z}T$ symmetry of the system, which we explain in the following remark:
\begin{rem}[Euler number]
To illustrate ideas, we assume we are close to a simple magic $\alpha$ and define the complex vector bundle of rank $2$:
\[ \begin{split} \mathcal E_0:= \{ [k,\phi]_{\tau} \in (\CC \times L^2_0(\CC/\Lambda;\CC^4)) / \sim_{\tau}: \phi \in \indic_{E_{\pm 1}(\alpha,k)}(H_k(\alpha)) \} \\
(k,\phi) \sim_{\tau} (k',\phi') \Leftrightarrow \exists p \in \Gamma^* , k'=k+p , \phi' = \tau(p) \phi,
\end{split}\]
where $\tau(p)\phi = e^{i \Re(p \bar z ) } \phi.$
Then, we consider the real subbundle 
\[\mathcal E = \{ \varphi \in \mathcal E_0 ; \mathcal{PT} \varphi= \varphi\}. \]
The PT symmetry is defined as 
\[ \mathcal{PT}:=\begin{pmatrix} 0 & \mathscr Q \\ \mathscr Q & 0 \end{pmatrix} \text{ with } \mathscr Qv(z)=\overline{v(-z)}.\]
It is a real vector bundle of rank $2$ since for all $\varphi_1,\varphi_2 \in \mathcal E$
\[ \langle \varphi_1,\varphi_2 \rangle =\overline{ \langle  \mathcal{PT} \varphi_1, \mathcal{PT} \varphi_2 \rangle} = \overline{\langle \varphi_1,\varphi_2 \rangle}.\]
Similar to how the Chern number measures the triviality of complex vector bundles, it is the Euler number that measures the trivial of the real vector bundle $\mathcal E$. In our case, we can interpret $\mathcal E$ as a complex line bundle with Chern number $-1$. This is explained in more detail in the last section of \cite{FBM}. In this sense, the non-zero Chern numbers above are an effect of the symmetry of the system. 
\end{rem}

One also readily verifies the usual properties of Chern characters for our $\Omega_i$, see, for instance, \cite[Lemma 3.1]{GKS}, \cite{BES94}:
\begin{prop}
\label{prop:auxiliary}
Let $P$ be an orthogonal projection satisfying \eqref{eq:decay}, then 
\begin{enumerate}
\item  $\vert \Omega_i(P)\vert \lesssim_{\kappa,\xi} K_P^2.$
\item Let $s \in \mathbb R$ and define $\hat \Theta^{(s)}_{j,i}(t) :=\pi_i  \Theta_j(t-s) $, then 
\[ \Omega_i^{r,s}(P):=\tr(P[[P,\hat \Theta_{1,i}^{(s)}],[P,\hat \Theta_{2,i}^{(r)}]]) \text{ for }r,s \in \RR.\]
In particular, 
\begin{equation}
\label{eq:covariance_rel}
\Omega_i^{r,s} = \Omega_i.
\end{equation}
\item Let $P,Q$ be two orthogonal projections, each satisfying \eqref{eq:decay}, such that $PQ = QP = 0$, then \[ \Omega_i(P+Q) = \Omega_i(P)+\Omega_i(Q).\]
\end{enumerate}
\end{prop}
\begin{proof}
The first property follows readily from the combination of \eqref{eq:HS_Bound} with the argument for the full Chern number in \cite[Lemma 3.1 (i)]{GKS}. 
If $\Re(z_0)\Re(z_1)>0$ then $\chi_{z_0} [P,\Lambda_1]\chi_{z_1}=0$. Thus, we may restrict ourselves to $\Re(z_0)\Re(z_1)\le 0$, we have $2\vert \Re(z_0-z_1)\vert^{\xi} \ge \vert \Re(z_0)\vert^{\xi} + \vert \Re(z_1)\vert^{\xi}. $
The second property follows from a direct computation; see \cite[Lemma 3.1 (ii)]{GKS}. 

The last property, shown in \cite[Lemma $8$]{BES94} for Chern numbers, follows from $P[Q, \hat \Theta_i] = -P\hat \Theta_i Q$ and evaluating \eqref{eq:definition} since one finds for the cross-terms 
\[ \begin{split} 
 \tr\Big(-P \hat \Theta_1 Q\hat \Theta_2 +P \hat \Theta_1 Q \hat \Theta_2 P -Q \hat \Theta_1 P \hat \Theta_2 + Q \hat \Theta_1 P \hat \Theta_2 Q  \Big)=0.
 \end{split}\]
This implies that
\[\begin{split} \Omega_i(P+Q):&=\tr( [(P+Q)\hat \Theta_1(i) (P+Q) , (P+Q) \hat \Theta_2(i) (P+Q)])  \\
&=\Omega_i(P)+ \Omega_i(Q).
\end{split} \]
\end{proof}
We also want to mention reference \cite[Sec.6]{ASS} showing full details on how to obtain the second point.

The independence of switch functions $\hat \Theta^{(s)}_{j,i}$ in Prop. \ref{prop:auxiliary} implies that $\Omega_i$ is an almost surely constant quantity
\begin{equation}
\label{eq:constant}
 \Omega_i(P)=\mathbf E \Omega_i(P) \text{ for }\mathbf P\text{-almost surely}.
 \end{equation}

The purpose of the first and last point in Prop. \ref{prop:auxiliary} is to conclude that in regions of $\operatorname{SUDEC}$, cf. Definition \ref{defi:Sudec}, all $\Omega_i$ vanish.
\begin{prop}
\label{prop:SUDEC}
Let $H_{\lambda}$ exhibit $\operatorname{SUDEC}$ in an interval $J$, then for all closed $I \subset J$ we have 
\[ \Omega_i(\indic_{I}(H_{\lambda}))=0\text{ for }\mathbf P\text{-almost surely}.\] 
\end{prop}
\begin{proof}
Let $M \subset \mathbb N$ be a (finite or infinite) enumeration (counting multiplicities) of all point spectrum of $H_{\lambda}$. We can then write the spectral projection as
\[ \indic_{I}(H_{\lambda}) = \sum_{m \in M} P_{m}\]
where $P_m$ are rank one projections. In addition, we have $K_P:=\sum_{m\in M} \alpha_{m}$ where $\alpha_{m}$ are defined in \eqref{eq:SUDEC}.
Using the third item in Prop. \ref{prop:auxiliary} we then have for any $\{1,...,N\} \subset M$
\[\Omega_i(\indic_{I}(H_{\lambda})) =\sum_{m =1}^N \underbrace{\Omega_i(P_{m})}_{=0} +\Omega_i\Bigg(\sum_{m \in M\setminus \{1,..,N\}} P_{m} \Bigg)= \Omega_i\Bigg(\sum_{m \in M\setminus \{1,..,N\}} P_{m} \Bigg). \]
By the first item in Prop. \ref{prop:auxiliary}, we find that as we let $N$ go to infinity or when it is equal to $\vert M \vert$, if $M$ is finite, that we obtain $\Omega_i(\sum_{m \in M\setminus \{1,..,N\}} P_{m} ) \to 0.$
\end{proof}

\smallsection{Acknowledgements:} We thank Jie Wang for suggesting the relevance of different disorder types for TBG and Maciej Zworski for initial discussions on the project. M. Vogel was partially funded by the Agence Nationale de la Recherche, through the project ADYCT (ANR-20-CE40-0017). I. Oltman was jointly funded by the National Science Foundation Graduate Research Fellowship under grant DGE-1650114 and by grant DMS-1901462.

\smallsection{Data availability:}  No new data were created or analysed in this study.
\end{appendix}


\begin{thebibliography}{0}
\bibitem[AALR79]{AALR} E. Abrahams, P.W. Anderson, D. C. Licciardello and T. V. Ramakrishnan, \emph{Scaling Theory
of Localization: Absence of Quantum Diffusion in Two Dimensions}, Physical
Review Letters, Vol. 42, No. 10, 1979, pp. 673-676.
\bibitem[AS19]{An} N.~Anantharaman, M.~Sabri, {\em Quantum ergodicity on graphs: From spectral to spatial delocalization},
Ann. of Math. (2) 189(3): 753-835 (2019).
\bibitem[AW13]{AW} M.~Aizenman, S.~Warzel, {\em Resonant delocalization for random Schrödinger operators on tree graphs},
J. Eur. Math. Soc 15 (4), 1167-1222, (2013).
\bibitem[ASS94]{ASS} J. ~Avron, R. ~Seiler, and B.~Simon, \emph{Charge Deficiency, Charge Transport
and Comparison of Dimensions}, Commun. Math. Phys. 159, 399-422, (1994).
\bibitem[An58]{And} P. ~Anderson, \emph{Absence of diffusion in certain random lattices.} Phys. Rev. 109, 1492-
1505, (1958).
\bibitem[BCZ19]{BCZ19} J.~-M.~Barbaroux, H.~Cornean, and S.~Zalczer, \emph{Documenta mathematica, 24:65-93}, (2019).

\bibitem[BES94]{BES94} J. Bellissard, A. van Elst, and H. Schulz‐-Baldes, \emph{The noncommutative geometry of the quantum Hall effect},
J. Math. Phys. 35, 5373 (1994).
\bibitem[B88]{B88} J. Bellissard, \emph{Ordinary quantum Hall effect and non-commutative cohomology}. In: Proc. of
the Bad Schandau Conference on Localization, 1986, eds. P. Ziesche, W. Weller, Teubner
Texte Phys. 16, Leipzig: Teubner-Verlag, (1988).

\bibitem[Be*21]{suppl} S.~Becker, M.~Embree, J.~Wittsten and M.~Zworski,
{\em Spectral characterization of magic angles in twisted bilayer graphene,} Phys. Rev. B {\bf 103}, 165113, (2021).
\bibitem[Be*22]{beta}  S.~Becker, M.~Embree, J.~Wittsten and M.~Zworski,
{\em Mathematics of magic angles in a model of twisted bilayer graphene},
 Probab. Math. Phys. {\bf 3},69--103, (2022) 
 \bibitem[BH22]{BH} S.~Becker, R.~Han, {\em Density of States and Delocalization for Discrete Magnetic Random Schrödinger Operators},
International Mathematics Research Notices 2022 (17), 13447-13504.
\bibitem[BHZ22]{bhz2} S.~Becker, T.~Humbert and M.~Zworski,
{\em Fine structure of flat bands in a chiral model of magic angles}, \arXiv{2208.01628}, (2022).
\bibitem[BHZ22b]{bhz3} S.~Becker, T.~Humbert and M.~Zworski,
{\em Integrability in the chiral model of magic angles}, \arXiv{2208.01620}, (2022).
\bibitem[BHZ23]{bhz4} S.~Becker, T.~Humbert and M.~Zworski,
{\em Degenerate flat bands in twisted bilayer graphene}, \arXiv{2306.02909}, (2023).
\bibitem[BOV24]{bov} S.~Becker, I.~Oltman and M.~Vogel \emph{Absence of small magic angles for disordered tunneling potentials in twisted bilayer graphene}, \arXiv{2402.12799},(2024).
\bibitem[BQTWY24]{FBM} S.~Becker, S.~Quinn, Z.~Tao, A.~Watson and M.~Yang, \emph{Dirac cones and magic angles in the Bistritzer--MacDonald TBG Hamiltonian}, \arXiv{2407.06316}, (2024). 
\bibitem[BZ23]{dynip} S.~Becker and M.~Zworski,
{\em Dirac points for twisted bilayer graphene with in-plane magnetic field}, \arXiv{2303.00743}, (2023). 
\bibitem[BM87]{Berry} M.V.~Berry and R.J.~Mondragon, \emph{Neutrino billiards:time-reversal symmetry-breaking without magnetic fields.}Proc. Roy. Soc. London Ser.A,412(1842):53–74, (1987).
\bibitem[BiMa11]{BM11} R.~Bistritzer and A.~MacDonald, {\em Moir\'e bands in twisted double-layer graphene}. PNAS, {\bf 108}, 12233--12237, (2011).
\bibitem[CH94]{CH94} J. M. Combes and P. D. Hislop, \emph{Localization for some continuous, random Hamiltonians in d dimensions}, J. Funct. Anal. 12, 149-180, (1994).
\bibitem[CHK07]{CHK} J.-M. Combes, P. D. Hislop, and F. Klopp, \emph{ An optimal Wegner estimate and its application to the global continuity of the integrated density of states for random Schrödinger operators}, Duke Math. J. 140(3), 469–498, (2007).
\bibitem[CHN01]{CHN01} J.-M. Combes, P. D. Hislop, and S. Nakamura, \emph{The $L^p$-theory of the spectral shift function, the Wegner estimate, and the integrated density of states for some random operators}, Comm. Math. Phys. 218, no. 1, 113–130,(2001).
\bibitem[CHK03]{CHK03} J.-M. Combes, P. D. Hislop, and F. Klopp, \emph{H\"older Continuity of the Integrated Density of States for Some Random Operators at All Energies}, International Mathematics Research Notices, No. 4,  (2003).
\bibitem[CMM19]{CMM19}H.D.Cornean, D. Monaco,and M. Moscolari, \emph{Parseval Frames of Exponentially Localized Magnetic Wannier Functions.} Commun. Math. Phys. 371, 1179-1230, (2019).
\bibitem[CRQ23]{CRQ23}V. Crépel, N. Regnault, R. Queiroz, \emph{The chiral limits of moir\'e semiconductors: origin of flat bands and topology in twisted transition metal dichalcogenides homobilayers}, \arXiv{2305.10477}, (2023).

\bibitem[CFKS87]{CFKS} H.L Cycon, R.G. Froese, W. Kirsch, W., B. Simon, \emph{Schr\"odinger Operators. Lecture Notes
in Physics}, Berlin-Heidelberg-New York: Springer Verlag, (1987).
\bibitem[ET05]{ET05} M. Embree, L. Trefethen, \emph{Spectra and Pseudospectra: The Behavior of Nonnormal Matrices and Operators}, Princeton University Press, (2005).
\bibitem[DS10]{DS10} M. Dimassi and J.Sjostrand, (1999). \emph{Spectral Asymptotics in the Semi-Classical Limit}. Cambridge: Cambridge University Press.
\bibitem[GK01]{GK01} F.~Germinet, A.~Klein, \emph{Bootstrap Multiscale Analysis and Localization in Random Media.} Commun. Math. Phys. 222, 415–448, (2001).
\bibitem[GK03]{GK02}F.~Germinet, A.~Klein, \emph{Explicit finite volume criteria for localization in continuous
random media and applications.} Geom. Funct. Anal. 13 1201-1238, (2003).
\bibitem[GK04]{GK04} F.~Germinet, A.~Klein, \emph{A characterization of the Anderson metal-insulator transport
transition.} Duke Math. J. 124, 309-351, (2004).
\bibitem[GK06]{GK06} F.~Germinet, A.~Klein, \emph{New characterizations of the region of complete localization for random Schr\"odinger operators,} J. Statist. Phys. 122, 73–94, (2006).
\bibitem[GKS07]{GKS} F.~Germinet, A.~Klein, and J.~Schenker, {\emph{Delocalization in random Landau Hamiltonians}}, Annals of Mathematics, 166, 215-244, (2007). 
\bibitem[GKS09]{GKS2} F.~Germinet, A.~Klein, and J.~Schenker, {\emph{Quantization of the Hall conductance and delocalization in ergodic Landau Hamiltonians}}, Reviews in Mathematical Physics, Vol. 21, No. 08, pp. 1045-1080, (2009).
\bibitem[GKM09]{GKM}F.~Germinet, A.~Klein, and B.~Mandy, {\emph{Delocalization for random Landau Hamiltonians with unbounded random variables}}, Contemporary Mathematics, Volume 500, (2009).
\bibitem[JSS03]{JSS} S. Jitomirskaya, H. Schulz-Baldes, G. Stolz, \emph{Delocalization in random polymer models}, Commun. Math. Phys. 233, 27-48, (2003).
\bibitem[K89]{Kirsch} W.~Kirsch, \emph{Random Schr\"odinger operators}, in "Schr\"odinger Operators" (H. Holdenand A. Jensen, Eds.),Lecture Notes in Physics,Vol. 345, Springer-Verlag, Berlin New York, (1989).
\bibitem[KM82]{KM82} Kirsch and F. Martinelli, \emph{On the spectrum of Schrödinger operators with a random
potential}, Commun. Math. Phys. 85, 329-350 (1982).
\bibitem[KS80]{KS80} H. Kunz, B. Souillard, \emph{Sur le spectre des opérateurs aux différences finies aléatoires.}
Commun. Math. Phys. 78, 201-246 (1980).
\bibitem[Li21]{Li21} T.~Li, S.~Jiang, B.~Shen. et al. \emph{Quantum anomalous Hall effect from intertwined moiré bands.} Nature 600, 641–646, (2021).
\bibitem[LS21]{LS21} J.~Lu and K.~D Stubbs, {\emph{Algebraic localization of Wannier functions implies Chern triviality in non-periodic insulators}}, \arXiv{2107.10699}, (2021).
\bibitem[Ka55]{Ka55}T. Kato, \emph{Notes on Projections and Perturbation Theory Technical Report}
No. 9 Univ. Calif., Berkley (1955)).
\bibitem[Ka80]{Kato} T. Kato, \emph{Perturbation Theory for Linear Operators}, Corrected second edition, Springer, Berlin, 1980.   
\bibitem[Kl98]{Kl98} A. Klein,\emph{Extended
states in the Anderson model on the Bethe lattice}, Advances in Mathematics
133 (1), 163-184, 1998.   
\bibitem[MMP]{MMP}G.~Marcelli, M.~Moscolari, and G.~Panati, \emph{Localization implies Chern triviality in non periodic insulators}, \arXiv{2012.14407}, (2020).
\bibitem[MPPT18]{MPPT} D.~Monaco, G.~Panati, A.~Pisante, S.~Teufel, \emph{Optimal Decay of Wannier functions in Chern and Quantum Hall Insulators.} Commun. Math. Phys. 359, 61-100, (2018).
\bibitem[N01]{N} S.~ Nakamura, \emph{A Remark on the Dirichlet-Neumann Decoupling and the Integrated Density of States}, Journal of Functional Analysis 179,136152, (2001).
\bibitem[Pa80]{Pa} L. Pastur, \emph{Spectral properties of disordered systems in one-body approximation.} Commun.
Math. Phys. 75, (1980). 
\bibitem[Si77]{Simon} B.~Simon, \emph{Notes on infinite determinants of Hilbert space operators}, Adv. in Math. 24, 244-273, (1977).
\bibitem[Si00]{Si00} B.~Simon, \emph{Schrödinger Operators in the Twenty-First Century}. Mathematical Physics. Imperial College London. pp. 283–288, (2000). 
\bibitem[Sj89]{Sj} J. Sj\"ostrand, \emph{Microlocal analysis for periodic magnetic Schr\"odinger equation and related questions, in Microlocal Analysis and Applications}, J.-M. Bony, G. Grubb, L. H\"ormander, H. Komatsu and J. Sj\"ostrand eds. Lecture Notes in Mathematics 1495, Springer, (1989).
\bibitem[SN47]{SN47} B. Sz.-Nagy, \emph{Perturbations des transformations
autoadjointes dans l’espace de Hilbert}, Comment. Math. Helv., 19, pp. 347-366, (1947).
\bibitem[SV19]{SV19} E.~Stockmeyer, S.~Vugalter, \emph{Infinite mass boundary conditions for Dirac operators.} J. Spectr. Theory 9, no. 2, pp. 569–600, (2019).
\bibitem[T92]{T92} B. Thaller, \emph{The Dirac Equation}, Springer-Verlag, (1992). 
\bibitem[TKV19]{magic}
G. Tarnopolsky,  A.J. Kruchkov and A. Vishwanath,
\emph{Origin of magic angles in twisted bilayer graphene},
Phys. Rev. Lett. 122, 106405, (2019).
\bibitem[TaZw23]{notes} Z.~Tao and M.~Zworski, 
{\em PDE methods in condensed matter physics,} Lecture Notes, 2023, \\ \url{https://math.berkeley.edu/~zworski/Notes_279.pdf}.

\bibitem[W95]{W} W.-M. Wang, \emph{Asymptotic expansion for the density of states of the magnetic Schrödinger operator with a random potential.} Commun. Math. Phys. 172, 401--425, (1995).
\bibitem[Y92]{Y92} D. R. Yafaev, \emph{Mathematical Scattering Theory.} General Theory, Translations of Mathematical Monographs, vol. 105, American Mathematical Society, Rhode Island, (1992).  

\end{thebibliography}
\end{document}